\documentclass[11pt,a4paper]{article}
\usepackage{amsmath,amssymb,amsfonts,amsthm}
\usepackage{stmaryrd}
\usepackage{cite}
\usepackage{color}
\newtheorem{thm}{Theorem}[section]
\newtheorem{cor}{Corollary}[section]
\newtheorem{lem}{Lemma}[section]
\newtheorem{prop}{Proposition}[section]
\theoremstyle{definition}
\newtheorem{defn}{Definition}[section]
\theoremstyle{remark}

\numberwithin{equation}{section}

\newcommand{\be}{\begin{equation}}
\newcommand{\ee}{\end{equation}}
\newcommand{\bea}{\begin{eqnarray}}
\newcommand{\eea}{\end{eqnarray}}
\newcommand{\bse}{\begin{subequations}}
\newcommand{\ese}{\end{subequations}}
\def \b#1{\bar{#1}}
\def \t#1{\tilde{#1}}
\def \deg{\mathrm{deg\,}}
\def \j{{\llbracket}}
\def \k{{\rrbracket}}

\title{Integrable properties of the differential-difference Kadomtsev-Petviashvili hierarchy
and continuum limits}

\author{Wei Fu$^{a}$, ~~Lin Huang$^{b}$, ~~K.M. Tamizhmani$^{c}$, ~~
Da-jun Zhang$^{a}$\footnote{Corresponding author. E-mail address: djzhang@staff.shu.edu.cn}\\
{\small \it $^{a}$Department of Mathematics, Shanghai University, Shanghai 200444, P.R.China} \\
{\small \it $^{b}$School of Mathematical Sciences, Fudan University, Shanghai 200433, P.R.China} \\
{\small \it $^{c}$Department of Mathematics, Pondicherry University, Puducherry  605014, India}
}
\date{\today}

\begin{document}

\maketitle

\begin{abstract}
The paper reveals clear links between the differential-difference Kadomtsev-Petviashvili hierarchy
and the (continuous) Kadomtsev-Petviashvili hierarchy,
together with their symmetries, Hamiltonian structures and conserved quantities.
They are connected through a uniform continuum limit.
We derive isospectral and non-isospectral differential-difference Kadomtsev-Petviashvili flows through Lax triads,
where  the spatial variable $\bar{x}$ is looked as a new  independent variable
that is completely independent of the temporal variable $\bar{t}_1$.
Such treatments not only enable us to derive the master symmetry as one of integrable non-isospectral flows,
but also provide
simple representations for both isospectral and non-isospectral differential-difference Kadomtsev-Petviashvili  flows
in terms of zero curvature equations.
The obtained flows generate a Lie algebra with respect to Lie product $\j\cdot,\cdot\k$,
which further leads to two sets of symmetries for the isospectral differential-difference Kadomtsev-Petviashvili hierarchy,
and the symmetries generate a Lie algebra, too.
Making use of
the recursive relations of the flows, symmetries and Noether operator we derive
Hamiltonian structures for both isospectral and non-isospectral
differential-difference Kadomtsev-Petviashvili hierarchies.
The Hamiltonians generate a Lie algebra with respect to Poisson bracket $\{\cdot,\cdot\}$.
We then derive two sets of conserved quantities for the whole
isospectral differential-difference Kadomtsev-Petviashvili hierarchy
and they also generate a Lie algebra. All these obtained algebras have same basic structures.
Then, we provide a continuum limit which is different from Miwa's transformation.
By means of defining \textit{degrees} of some elements with respect to the continuum limit,
we prove that the differential-difference Kadomtsev-Petviashvili hierarchies together with their Lax triads,
zero curvature representations and integrable properties
go to their continuous counterparts in the continuum limit.
Structure deformation of Lie algebras in the continuum limit is also explained.

\vskip 6pt
\noindent{\textbf{Keywords:}} differential-difference Kadomtsev-Petviashvili hierarchy,
symmetries, Hamiltonian structures, conserved quantities, continuum limit.\\
\textbf{MSC 2010}: 37K05, 37K10, 37K30\\
\textbf{PACS}: 02.30.Ik
\end{abstract}
\vskip 20pt

\section{Introduction}

It is well known that the Kadomtsev-Petviashvili (KP) equation
\be
u_{t}=\frac{1}{4}u_{xxx}+3uu_x+\frac{3}{4}\partial^{-1}_x u_{yy}
\label{KP}
\ee
acts as a typical role in (2+1)-dimensional integrable systems.
This equation together with its bilinear form
is an elementary model in the celebrated Sato's theory \cite{OSTT-PTPS-1988,DJM-book-2000}.
The KP equation itself as well as its integrable characteristics,
such as infinitely many symmetries and conserved quantities,
can be derived from a pseudo-differential operator\footnote{Detailed definition of $L$ is given in next section.}
\cite{Dickey-book-2003,MSS-JMP-1990,CXZ-CSF-2003},
\be
L=\partial+ u_2\partial^{-1}+u_3\partial^{-2}+\cdots.
\label{L}
\ee
The operator can also generate a KP hierarchy\cite{OSTT-PTPS-1988,CXZ-CSF-2003}.
Most of (1+1)-dimensional Lax integrable systems have their own recursion operators,
while for (2+1)-dimensional systems it is quite rare to see that.
However, the KP hierarchy does have a recursive structure
which is expressed either through a recursion operator\cite{SF-CMP-1988-I} or through a master symmetry
 together with Lie product\cite{OF-PLA-1982}.
By means of the recursive structure, a KP hierarchy was built, and
symmetries, Hamiltonian structures and conserved quantities of the whole
isospectral KP hierarchies were generated\cite{Case-PNAS-1984,Case-JMP-1985,OF-PLA-1982,CLL-PD-1983,CLB-JPA-1988}.
In fact,
the KP hierarchy constructed in \cite{OF-PLA-1982} by using the recursive structure
and the KP hierarchy derived from the pseudo-differential operator \eqref{L} are same.

The differential-difference Kadomtsev-Petviashvili (D$\Delta$KP) equation reads
\begin{equation}
\b u_{\b t}=(1+2\Delta^{-1})\b u_{\b x\b x}-2h^{-1}\b u_{\b x}+2\b u \b u_{\b x},
\label{DDKP}
\end{equation}
with one discrete independent variable $n$ and two continuous ones $\b x$ and $\b t$,
where the operator $\Delta$ is defined by $\Delta f(n)=f(n+1)-f(n)$ and $h$ is a spacing parameter of $n$.
This equation is first derived through a discretization of the Sato's  theory \cite{DJM-JPSJ-1982-II}.
The discretization is also known as Miwa's transformation\cite{Miwa-PJA-1982}.
Based on the transformation it is quite natural to get bilinear identities with discrete exponential functions,
from which one can derive bilinear equations with discrete variables\cite{DJM-JPSJ-1982-I,DJM-JPSJ-1982-II,DJM-JPSJ-1983-III,DJM-JPSJ-1983-IV,DJM-JPSJ-1983-V}.
However, since Miwa's transformation does not keep the original continuous dispersion relation,
for a integrable discrete equation it is hard from the first glance to find the correspondence to a continuous counterpart.
It is first shown in \cite{KVT-CSF-1997} that the D$\Delta$KP equation is related to the following pseudo-difference operator
\be
\bar{L}=\Delta+ \b u_0 +\b u_1\Delta^{-1}+\b u_2\Delta^{-2} +\cdots,
\label{b-L}
\ee
with $\b u_0=\b u$ and $\b t_1=\b x$\cite{KVT-CSF-1997}.
By using the above pseudo-difference operator
some integrable properties of the D$\Delta$KP equation,
such as symmetries and conservation laws, were investigated\cite{KVT-CSF-1997,K-DOC-1998,SZZC-MPLB-2010,Zhang-JSU-2005}.

In this paper, for the D$\Delta$KP equation \eqref{DDKP}
we will first investigate the recursive structure of the D$\Delta$KP hierarchy.
To do that we need to introduce a master symmetry (cf.\cite{Fuch-PTP-1983}).
Usually master symmetries are related to time-dependent spectral parameters
and can be derived from spectral problems as non-isospectral flows.
Since isospectral  and non-isospectral D$\Delta$KP flows are simultaneously
considered and they are related to the same spectral problem, we can not take $\b t_1=\b x$ any longer
and we have to consider $\b x$ as a new independent variable.
Consequently, we use a Lax triad rather than a Lax pair to derive the D$\Delta$KP hierarchies
and it turns out that this works.

In the paper our plan is the following.
After introducing necessary notations in Sec.2, we will revisit the KP hierarchy in Sec.3.
We will derive isospectral and non-isospectral KP flows via Lax triad approach.
The approach provides simple zero curvature representations for these flows,
by which one can easily obtain a Lie algebra of the flows.
The basic structure of the algebra  indicates a recursive relation for both isospectral and non-isospectral KP flows.
Integrable properties of the isospectral KP hierarchy,
such as symmetries, Hamiltonian structures and conserved quantities
and Hamiltonian structures of non-isospectral KP hierarchy, are also listed out in this section as the known results in literature.
Next, in Sec.4, we focus on the D$\Delta$KP hierarchy.
By Lax triad approach we derive isospectral and non-isospectral D$\Delta$KP flows
and their basic algebraic structure.
The structure can be used to generate infinitely many symmetries for the isospectral D$\Delta$KP hierarchy
as well as provides a recursive relation of flows.
Then we will investigate their Hamiltonian structures and conserved quantities for the isospectral D$\Delta$KP hierarchy,
and Hamiltonian structures for the non-isospectral D$\Delta$KP hierarchy.
Finally in Sec.5, by means of continuum limit we will discuss possible connections
between the KP hierarchies and D$\Delta$KP hierarchies together with their Lax triads and integrability characteristics.

\section{Basic notions}

A pseudo-differential operator $L$ is defined as
\be
L=\partial+u_2\partial^{-1}+u_3\partial^{-2}+\cdots+u_{j+1}\partial^{-j}+\cdots,
\label{KP:L}
\ee
where $\partial\doteq\partial_x$, $\partial \partial^{-1}=\partial^{-1}\partial=1$
and $u_j=u_j(x,y,\mathbf{t})$ with $\mathbf{t}=(t_1, t_2,\cdots)$.
$\partial^{s}$ obeys the Leibniz rule
\be\label{KP:Leibnitz}
\partial^{s} f=\sum^{\infty}_{i=0} \mathrm{C}^{i}_{s}(\partial^{i}f)\partial^{s-i},~~~s\in \mathbb{Z},
\ee
where
\be
\mathrm{C}^{i}_{s}=\frac{s(s-1)(s-2)\cdots(s-i+1)}{i!}.
\label{def:Cij}
\ee
We suppose $\{u_j\}$ belong to a rapidly-decreasing function space $\mathcal{S}$,
and introduce a set
\[\mathcal{F}=\{f=f(u)|u=u(x,y,\mathbf{t})\in \mathcal{S} ~ \mathrm{and} ~ f(u)|_{u=0}=0\}.\]
The inner product $(\cdot,\,\cdot)$ on $\mathcal{F}$ is taken as
\be
(f,\,g)=\int_{-\infty}^{+\infty}\int_{-\infty}^{+\infty}f(u)g(u)\,\mathrm{d}x \mathrm{d}y,~~\forall f,g\in \mathcal{F}.
\label{def:inn prod}
\ee
A second product $\llbracket\cdot,\cdot\rrbracket$ on $\mathcal{F}$ is defined as
\be
\llbracket f, g\rrbracket=f'[g]-g'[f], ~~\forall f,g\in \mathcal{F},
\label{def:Lie product}
\ee
where
\be
f'[g]= \frac{d}{d\varepsilon}f(u+{\varepsilon}g)\bigr|_{\varepsilon=0}
\label{def:G deriv}
\ee
is  the G\^ateaux derivative of $f$ in direction $g$ w.r.t. $u$.

For a functional $H=H(u)$ and a function $f\in \mathcal{F}$, if
\be
H'[g]=(f,g),~~~\forall g\in\mathcal{F},
\label{def:grad}
\ee
then $f$ is called the functional derivative or gradient of $H$, and $H$ is called the potential of $f$.
Such an $f$ is usually denoted by $\frac{\delta H}{\delta u}$ or $\mathrm{grad}~H$.
\begin{prop}\label{P:2-1}\cite{FF-PD-1981}
$f\in \mathcal{F}$ is a gradient function if and only if $f'$
is a self-adjoint operator in terms of the inner product \eqref{def:inn prod}, i.e. $f'^*=f'$.
The corresponding  potential $H$ can be given by
\be
H=\int_{0}^{1}(f(\lambda u),u)\mathrm{d}\lambda.
\label{def:pot}
\ee
\end{prop}

For a given evolution equation
\be
u_t=K(u),
\label{def:evo eq}
\ee
$\kappa=\kappa(u)\in \mathcal{F}$ is a symmetry of the above equation if
\be
\kappa_t=K'[\kappa]
\label{def:sym}
\ee
holds for all of $u$ solving \eqref{def:evo eq}.
\eqref{def:sym} is alternately written as
\be
\tilde{\partial}_t\kappa=\llbracket K,\kappa\rrbracket,
\ee
where
the operator $\tilde{\partial}_t\kappa$ stands for taking the derivative w.r.t. $t$ only explicitly contained in $\kappa$,
e.g. if $\kappa=t u_x+uu_{xx}$, then $\tilde{\partial}_t\kappa=u_x$.
Function $\gamma=\gamma(u)\in\mathcal{F}$ is called a conserved covariant of equation \eqref{def:evo eq} if\cite{FF-PD-1981}
\be
\gamma_t=-K'^*[\gamma]
\label{def:cc}
\ee
or
\be
-\tilde{\partial}_t\gamma =\gamma'[K]+K'^*[\gamma]
\ee
holds for all of $u$ solving \eqref{def:evo eq}.
Here $K'^*$ is the adjoint operator of $K'$  w.r.t. $(\cdot,\,\cdot)$.
Functional $I=I(u)$ is called a conserved quantity of equation \eqref{def:evo eq}
if
\be
\frac{\partial I}{\partial t}=0
\ee
holds for any $u$ solving \eqref{def:evo eq}.
Conserved quantities and conserved covariants are closely related to each other (cf.\cite{FF-PD-1981}).
One relation is
\begin{prop}\label{P:2-2-0}
If $\kappa(u)$ is a symmetry and  $\gamma=\gamma(u)$ is a conserved covariant of equation \eqref{def:evo eq},
then
\be
I=(\kappa(u),\gamma(u))
\label{k-gamma}
\ee
is a conserved quantity of  \eqref{def:evo eq}.
\end{prop}
\begin{proof}
Let us give the proof for completeness. In fact,
\begin{align*}
\frac{d I}{d t}& =(\kappa_t,\gamma)+(\kappa,\gamma_t)\\
&=(K'[\kappa],\gamma)+(\kappa,-{K'^*}[\gamma] )\\
&=(K'[\kappa],\gamma)+(-K'[\kappa],\gamma)\\
&=0.
\end{align*}
\end{proof}

\noindent
Another relation is
\begin{prop}\label{P:2-2}
Suppose that $\gamma=\gamma(u)$ is a gradient function and functional $I=I(u)$ is its potential
and $\left.\frac{\partial I}{\partial t}\right|_{u=0}=0$.
Then,  $I$ is a conserved quantity of equation \eqref{def:evo eq}
if and only if $\gamma$ is a conserved covariant of \eqref{def:evo eq}.
\end{prop}

\begin{proof}
Let us first prove the equality
\be
\tilde \partial_t (I'[g])= (\t \partial_t I)'[g]+I'[\t \partial_t g],~~ \forall g=g(u)\in \mathcal{F}.
\label{eq-pt}
\ee
Write $I=I(t,\{u^{(j)}\})$ where $u^{(j)}=\partial^j_x u$.
Then,
\[I'[g]=\sum_j\frac{\partial I}{\partial u^{(j)}} \partial^j_x g,\]
and
\[\tilde \partial_t (I'[g]) = \sum_j\frac{\partial (\t \partial_t I)}{\partial u^{(j)}} \partial^j_x g
                             +\sum_j\frac{\partial I}{\partial u^{(j)}} \partial^j_x (\t \partial_t g)
                          = (\t \partial_t I)'[g]+I'[\t \partial_t g],
\]
i.e. \eqref{eq-pt}.
Since $\gamma=\mathrm{grad}\,I$, i.e. $I'[g]=(\gamma, g)$,
we then have
\[\t \partial_t (\gamma, g)=(\t \partial_t I)'[g]+(\gamma, \t \partial_t g),\]
i.e.
\be
(\t \partial_t I)'[g]=(\t \partial_t \gamma,\, g).
\label{eq-pit}
\ee

Next, when $u$ satisfies equation \eqref{def:evo eq} and noting that
\[\frac{\partial I}{\partial t}=\t\partial_t I+ I'[u_t]=\t\partial_t I+ I'[K]=\t\partial_t I+ (\gamma, K),\]
for any $g\in \mathcal{F}$ we then have
\begin{align*}
\left(\frac{\partial I}{\partial t}\right)'[g]
& =(\t\partial_t I)'[g]+  (\gamma, K)'[g]\\
&= (\t\partial_t \gamma , g)+  (\gamma'[g], K)+(\gamma, K'[g])\\
&=(\t\partial_t \gamma , g)+  ({\gamma '}^* K,g)+({K'}^* \gamma,g)\\
&=(\t\partial_t \gamma +{\gamma '} K+{K'}^* \gamma,g),
\end{align*}
where we have made use of $\gamma'={\gamma'}^*$.
Thus it is clear that
if $I$ is a conserved quantity of equation \eqref{def:evo eq} then $\gamma$ is a conserved covariant of \eqref{def:evo eq},
and vise versa.
\end{proof}

Operator $\Gamma$ living on $\mathcal{F}$ is called a Noether operator of equation \eqref{def:evo eq},
if
\be
\Gamma_t=\Gamma K'^*+K'\Gamma,
\ee
or equivalently,
\be
\tilde \partial_t \Gamma+\Gamma'[K]-\Gamma K'^*-K'\Gamma=0.
\label{def:Noether}
\ee
$\Gamma$ maps conserved covariants of \eqref{def:evo eq} to its symmetries.
Operator $\theta$ living on $\mathcal{F}$ is called an implectic operator\cite{FF-PD-1981}
if it is skew-symmetric as well as satisfies the Jacobi identity
\be
(f,\theta'[\theta g]h)+(h,\theta'[\theta f]g)+(g,\theta'[\theta h]f)=0,~~~\forall  f,g,h\in \mathcal{F}.
\label{Jacobi}
\ee
The evolution equation \eqref{def:evo eq} has a Hamiltonian structure if it can be written in the form
\be
u_t=\theta \frac{\delta H}{\delta u},
\label{def:Hamiltonian}
\ee
where $\theta$ is an implectic operator.

Next we introduce a discrete independent variable $n$ to replace the continuous variable $x$.
The basic operation w.r.t. $n$ is a shift.
Here by $E$ we denote a shift operator  defined through
$E^j g(n)=g(n+j)$ for $j\in \mathbb{Z}$.
Besides, difference operator $\Delta=E-1$ is a
discrete analogue of differential operator $\partial_x$,
and  $\Delta^{-1}=(E-1)^{-1}$ is defined by  $\Delta\Delta^{-1}=\Delta^{-1}\Delta=1$.
$\Delta^s$ follows a discrete Leibniz rule,
\be
\Delta^{s} g(n)=\sum^{\infty}_{i=0} \mathrm{C}^{i}_{s}\,(\Delta^{i}g(n+s-i))\Delta^{s-i},~~~s\in \mathbb{Z},
\label{dKP:Leibnitz}
\ee
where $\mathrm{C}^{i}_{s}$ is defined as before.
For example, we have
\begin{subequations}
\begin{align}
& \Delta g(n)=g(n+1)\Delta+(\Delta g(n)),
\label{D-1}\\
& \Delta^2 g(n)=g(n+2)\Delta^2+2(\Delta g(n+1))\Delta+(\Delta^2 g(n)),
\label{D-2}\\
& \Delta^{-1} g(n)=g(n-1)\Delta^{-1}-(\Delta g(n-2))\Delta^{-2}+\cdots+(-1)^{j-1}(\Delta^{j-1} g(n-j))\Delta^{-j}+\cdots.
\label{D-(-1)}
\end{align}
\end{subequations}
Formula \eqref{dKP:Leibnitz} can be proved by using mathematical inductive method,
and we specify that it is also valid for negative integer $s$.

A pseudo-difference operator is defined as the following,
\be
\b{L}=h^{-1} \Delta+\b u_{0}+h\b u_{1}\Delta^{-1}+\cdots+h^{j}\b u_{j}\Delta^{-j}+\cdots,
\label{dKP:L}
\ee
where $\b u_j=\b u_j(n,\b x,\b {\mathbf{t}})$ with $\b {\mathbf{t}}=(\b t_1,\b t_2,\cdots)$,
and $h$ acts as a lattice spacing parameter of $n$-direction.

As in continuous case, here we suppose $\{\b u_j\}$ belong to a rapidly-decreasing function space $\b{\mathcal{S}}$,
and also introduce a function set
\[\b{\mathcal{F}}=\{\b f=\b f(\b u)|\b u=\b u(n,\b x,\b {\mathbf{t}})\in \b{\mathcal{S}} ~ \mathrm{and} ~\b f(\b u)|_{\b u=0}=0\}.\]
The inner product in $\b{\mathcal{F}}$ is defined as
\be
(\b f(\b u),\b g(\b u))=\frac{h^2}{2}\sum_{n=-\infty}^{\infty}\int_{-\infty}^{+\infty}\b f(\b u)\b g(\b u)\,\mathrm{d}\b x,~~
\forall \b f,\b g\in\mathcal{\b F}.
\label{dKP:inn prod}
\ee

Then, we can define the semi-discrete
counterparts of those notions and propositions for the continuous case described from \eqref{def:Lie product} up to \eqref{def:Hamiltonian}
by formally same formulae. We skip them here.

\section{The KP system}\label{S:3}

\subsection{The KP equation}
\label{S:3.1}

Let us first quickly review the traditional derivation of the KP equation (cf. \cite{OSTT-PTPS-1988,CXZ-CSF-2003}).
After this, we will revisit it via Lax triad approach in next subsection.

The isospectral KP hierarchy (corresponding to $\eta_{t_m}=0$) arise from the compatibility condition of the linear problems
\bse
\begin{align}
L\phi & =\eta\phi,\\
\phi_{t_m} & =A_m\phi,
\end{align}
\ese
i.e.
\be\label{KP:Lax eq}
L_{t_m}=[A_m,L]=A_m L-LA_m,
\ee
where $L$ is the pseudo-differential operator \eqref{KP:L} and
 $A_m=(L^m)_{+}$ is the differential part of $L^m$.
Explicit formulae of $A_m$ are given in \cite{Zhang-JPSJ-2003}.
The first few of $A_m$ are
\bse
\begin{align}
A_1&=\partial,\\
A_2&=\partial^2+2u_2,\\
A_3&=\partial^3+3u_2\partial+3u_3+3u_{2,x},\\
A_4&=\partial^4+4u_2\partial^2+(4u_3+6u_{2,x})\partial+4u_4+6u_{3,x}+4u_{2,xx}+6u_2^2.
\end{align}
\ese
From \eqref{KP:Lax eq} we have
\be\label{ut1}
u_{j,t_1}=u_{j,x},\quad(j=2,3,\cdots);
\ee
\bse\label{ut2}
\begin{align}
u_{2,t_2}&=2u_{3,x}+u_{2,xx},\\
u_{3,t_2}&=2u_{4,x}+u_{3,xx}+2u_2u_{2,x},\\
u_{4,t_2}&=2u_{5,x}+u_{4,xx}+4u_{2,x}u_3-2u_2u_{2,xx},\\
&~~\cdots\cdots;\nonumber
\end{align}
\ese
\bse\label{ut3}
\begin{align}
u_{2,t_3}&=3u_{4,x}+3u_{3,xx}+u_{2,xxx}+6u_2u_{2,x},
\label{ut3-1}
\\
u_{3,t_3}&=3u_{5,x}+3u_{4,xx}+u_{3,xxx}+6(u_2u_3)_x,\\
&~~\cdots\cdots;\nonumber
\end{align}
\ese
\bse\label{ut4}
\begin{align}
u_{2,t_4}&=4u_{5,x}+6u_{4,xx}+4u_{3,xxx}+u_{2,xxxx} +12(u_2u_3)_x+6(u_2u_{2,x})_x,\\
&~~\cdots \cdots.\nonumber
\end{align}
\ese

To derive the KP equation one first needs to set $t_1=x,~t_2=y$ and
next using \eqref{ut2} one can successfully express $u_3$ and $u_4$ by $u_2$
and then from equation \eqref{ut3-1}
one obtains the isospectral KP equation \eqref{KP}, i.e.
\be\label{iKP:eq}
u_{t_3}  =\frac{1}{4}u_{xxx}+3uu_x+\frac{3}{4}\partial^{-1}u_{yy}.
\ee

\subsection{Lax triad and the isospectral KP hierarchy }

Noting that the function $u$ in the KP equation \eqref{iKP:eq} depends on three independent variables $(x,y,t_3)$,
a Lax triad is actually needed for matching these three independent variables.
Later, we will also see when we derive a master symmetry as a non-isospectral flow
we can not take $t_2$ to be $y$ any longer and we have to consider $y$ and $t_2$ separately.
This also requires a triad rather than a pair.

For the whole KP hierarchy we need
\bse
\label{iKP:triad}
\begin{align}
& L\phi =\eta\phi,~~~\eta_{t_m}=0,\\
& \phi_{y} =A_2\phi,~~ A_2=\partial^2+2u_2,\\
& \phi_{t_m} =\hat{A}_m\phi, ~~~ m=1,2,\cdots,
\end{align}
\ese
where we suppose
\be
\hat{A}_m= \partial^m+\sum_{j=1}^{m}a_j\partial^{m-j},~~ \hat{A}_m|_{\mathbf{u}=\mathbf{0}}=\partial^m,
\ee
with $\mathbf{u}=(u_2,u_3,\cdots)$.
We leave the coefficients $\{a_j\}$ temporarily unknown.
The compatibility of \eqref{iKP:triad} reads
\bse
\label{iKP:com}
\begin{align}
& L_y  =[A_2,L],\label{iKP:com-1}\\
& L_{t_m}  =[\hat{A}_m,L],\label{iKP:com-2}\\
& A_{2,t_m}-\hat{A}_{m,y}+[A_2,\hat{A}_m]=0, ~~~ m=1,2,\cdots.\label{iKP:com-3}
\end{align}
\ese
Among the above compatibility conditions, \eqref{iKP:com-1} gives the relation \eqref{ut2} with
$y$ in place of $t_2$, which will be used to express $\{u_j\}_{j>2}$ by $u_2$,
as the following,
\bse
\begin{align}
u_3&=\frac{1}{2}(\partial^{-1}u_{2,y}-u_{2,x}),\\
u_4&=\frac{1}{4}(\partial^{-2}u_{2,yy}-2u_{2,y}+u_{2,xx}-2u_2^2),\\
u_5&=\frac{1}{8}(\partial^{-3}u_{2,yyy}-3\partial^{-1}u_{2,yy}+3u_{2,xy}-u_{2,xxx}\nonumber\\
&~~~~+12u_2u_{2,x}-8u_2\partial^{-1}u_{2,y}+4\partial^{-1}u_2u_{2,y}),\\
&~~~\cdots\cdots.\nonumber
\end{align}
\ese
The equation \eqref{iKP:com-2} plays the role to determine those unknowns $\{a_j\}$ of $\hat{A}_m$.
In fact\cite{Zhang-JPSJ-2003}, $\{a_j\}$ can be uniquely determined from \eqref{iKP:com-2}
and it turns out that  $\hat{A}_m$ is nothing but $A_m=(L^m)_{+}$.
The third equation \eqref{iKP:com-3}
provides the isospectral KP hierarchy
\be\label{iKP:hie}
u_{t_m}=K_{m}=\frac{1}{2}({A}_{m,y}-[A_2,{A}_m]),~~m=1,2,\cdots,
\ee
where we neglect the $\hat{}$ sign due to $\hat{A}_m=A_m=(L^m)_{+}$,
and we have taken $u_2=u$.
Let us write down the first four equations in the KP hierarchy:
\bse\label{iKP:flow}
\begin{align}
u_{t_1}& =K_1=u_x,\\
u_{t_2}& =K_2=u_y,\\
u_{t_3}& =K_3=\frac{1}{4}u_{xxx}+3uu_x+\frac{3}{4}\partial^{-1}u_{yy},\label{iKP:flow-3}\\
u_{t_4}& =K_4=\frac{1}{2}u_{xxy}+4uu_y+2u_x\partial^{-1}u_y+ \frac{1}{2}\partial^{-2}u_{yyy}.
\end{align}
\ese

\subsection{Lax triad and the non-isospectral KP hierarchy }

To derive a master symmetry we turn to the non-isospectral case in which we set
\be\label{nKP:eta}
\eta_{t_m}=\eta^{m-1},~~m=1,2,\cdots.
\ee
In this turn the Lax triad reads
\bse
\label{nKP:triad}
\begin{align}
& L\phi  =\eta\phi,\\
& \phi_{y} =A_2\phi, \\
& \phi_{t_m} =B_m\phi, ~~~ m=1,2,\cdots,
\end{align}
\ese
and the compatibility is
\bse
\label{nKP:com}
\begin{align}
& L_y  =[A_2,L],\label{nKP:com-1}\\
& L_{t_m}  =[B_m,L]+L^{m-1},\label{nKP:com-2}\\
& A_{2,t_m}-B_{m,y}+[A_2,B_m]=0, ~~~ m=1,2,\cdots,\label{nKP:com-3}
\end{align}
\ese
where we suppose  $B_m$ is an undetermined operator of the form
\be
B_m= \sum^{m}_{j=0} b_j \partial^{m-j}.
\ee
Checking the asymptotic results \eqref{nKP:com-2}$_{\mathbf{u}=\mathbf{0}}$ and \eqref{nKP:com-3}$_{\mathbf{u}=\mathbf{0}}$
respectively, one finds they  together give
the necessary asymptotic condition for $B_m$:
\be\label{nKP:B bc-1}
B_m|_{\mathbf{u}=\mathbf{0}}=2y\partial^m+x\partial^{m-1},~~m=1,2,\cdots.
\ee
We note that one can also add isospectral asymptotic terms, for example,
\be
 B_m|_{\mathbf{u}=\mathbf{0}}=2y\partial^m+x\partial^{m-1}+ \partial^{m-2}
\ee
when $m\geq 3$. This will lead to a non-isospectral flow combined by a isospectral flow $K_{m-2}$ and this does not
change the basic algebraic structure of the flows (see Sec.\ref{S:3.4}).

With the asymptotic condition \eqref{nKP:B bc-1} the operator $B_m$
can uniquely be determined from \eqref{nKP:com-2} and the first few of them are
\bse
\begin{align}
B_1 & = 2yA_1+x,\\
B_2 & = 2yA_2+xA_1,\\
B_3 & = 2yA_3+xA_2+(\partial^{-1}u_2),\\
B_4 & = 2yA_4+xA_3+(\partial^{-1}u_2)\partial+2(\partial^{-1}u_3),
\end{align}
\ese
where $A_j=(L^j)_+$ are defined as in Sec.\ref{S:3.1}.
Then, from \eqref{nKP:com-3} we have the non-isospectral KP hierarchy
\be\label{nKP:hie}
u_{t_m}=\sigma_{m}=\frac{1}{2}(B_{m,y}- [A_2, B_m]),~~m=1,2,\cdots,
\ee
and the first four equations are
\bse\label{nKP:flow}
\begin{align}
u_{t_1}  &=\sigma_{1}=2yK_{1},\\
u_{t_{2}}&=\sigma_{2}=2yK_{2}+xK_1+2u ,\\
u_{t_{3}}&=\sigma_{3}=2yK_{3}+xK_2+2\partial^{-1}u_{y}-u_x,\\
u_{t_{4}}&=\sigma_{4}=2yK_{4}+xK_{3}+u_{xx}+4u^{2}+u_{x}\partial^{-1}u+\frac{3}{2}\partial^{-2}u_{yy}-\frac{3}{2}u_y,
\end{align}
\ese
where $\{K_j\}$ are isospectral flows given in \eqref{iKP:hie},
and we have taken $u_2=u$.

$\{K_m\}$ and $\{\sigma_m\}$ are respectively called the isospectral and non-isospectral KP flows.
They are used to generate symmetries, Hamiltonians and conserved quantities for the isospectral KP hierarchy \eqref{iKP:hie}.
For these flows we have
\begin{prop}\label{P:3-1}
For the isospectral and non-isospectral KP flows $\{K_s\}$ and $\{\sigma_s\}$ we have
\bse\label{KP:zcr}
\begin{align}
K_s & =\frac{1}{2}(A_{s,y}-[A_2, A_s]),\label{iKP:zcr}\\
\sigma_s & =\frac{1}{2}(B_{s,y}-[A_2, B_s]),\label{nKP:zcr}
\end{align}
which are called zero curvature representations of the isospectral flow $K_s$
and  non-isospectral flow $\sigma_s$, respectively.
Here we specify the asymptotic data
\begin{align}
& K_s|_{u=0}=0,~ A_{s}|_{u=0}=\partial^{s},\label{KP:asym-a}\\
& \sigma_s|_{u=0}=0,~ B_{s}|_{u=0}=2y\partial^s+x\partial^{s-1},\label{KP:asym-b}
\end{align}
\ese
for $s=1,2,\cdots$.
\end{prop}

Besides, the isospectral flows $\{K_s\}$ can also be expressed in terms of the pseudo-differential operator $ L$.
\begin{prop}\label{P:3-2}
The isospectral flows $\{K_s\}$ defined by \eqref{iKP:zcr} can be expressed as
\be
K_s=\partial \,\underset{\partial}{\mathrm{Res\,}} L^s,
\label{K-Res}
\ee
where
\[\underset{k}{\mathrm{Res\,}}\biggl( \sum^{+\infty}_{j=-m}c_j k^j\biggr) =c_{-1},~~ (m\geq 1).\]
\end{prop}

\begin{proof}
From \eqref{iKP:zcr} we have
\begin{align*}
2K_s &= A_{s,y}-[A_2,A_s]\\
     &=[(L^s-(L^s)_{-})_y-[A_2, L^s-(L^s)_{-}]]_{0}\\
     &=[(L^s)_y-[A_2,L^s]-((L^s)_{-})_y+[A_2, (L^s)_{-}]]_{0}.
\end{align*}
Here $(L^s)_{-}=L^s-(L^s)_{+}$, and $(\,\cdot\,)_0$  means taking the constant part of the operator  $(\,\cdot\,)$.
Noting that \eqref{iKP:com-1} indicates $(L^s)_y-[A_2,L^s]=0$
we then have
\[2 K_s=[A_2, (L^s)_{-}]_{0}= 2\partial \,\underset{\partial}{\mathrm{Res\,}} L^s\]
and we finish the proof.
\end{proof}

\subsection{Algebra of flows, recursive structures and symmetries}
\label{S:3.4}

The KP flows $\{K_l\}$ and $\{\sigma_r\}$ generate  a Lie algebra w.r.t. the product
$\llbracket\cdot\, ,\,\cdot\rrbracket$  defined in \eqref{def:Lie product}.
This fact can be proved by using the zero curvature representations of these flows.
\begin{thm}
\label{thm:KP alg}
The KP flows $\{K_{l}\}$ and $\{\sigma_{r}\}$ span (or generate) a Lie algebra\footnote{By this we mean that
$\{K_{l}\}$ and $\{\sigma_{r}\}$ generate a linear space $\mathbf{X}=\Bigl\{\underset{j}{\sum} \alpha_jK_j
+\underset{j}{\sum} \beta_j \sigma_j,~ ~\alpha_j,\beta_j\in \mathbb{R}\Bigr\}$
which is closed w.r.t. the Lie product $\j\cdot,\cdot\k$. }
$\mathbf{X}$ with basic structure
\bse\label{KP:alg}
\begin{align}
&{\llbracket}K_{l},K_{r}{\rrbracket}= 0,\label{KP:alg-1}\\
&{\llbracket}K_{l},\sigma_{r}{\rrbracket}= l\,K_{s+r-2},\label{KP:alg-2}\\
&{\llbracket}\sigma_{l},\sigma_{r}{\rrbracket}= (l-r)\sigma_{l+r-2},\label{KP:alg-3}
\end{align}
\ese
where $l,r\geq 1$ and we set $K_0=\sigma_0=0$.
\end{thm}

We prove the theorem through the following two lemmas.
The first is
\begin{lem}\label{lem:KP-1}
For the function $X=X(u)\in \mathcal{F}$ and differential operator
$$ N=a_0\partial^m+a_1\partial^{m-1}+\cdots+a_{m-1}\partial+a_m,~~N|_{u=0}=0 $$
living on $\mathcal{F}$,
the equation
\be\label{KP:test eq}
2X-N_{y}+[A_2, N]=0
\ee
has only zero solution $X=0,~N=0$.
Here $A_2=\partial^2+2u$ where we have taken $u_2=u$.
\end{lem}

\begin{proof}
Comparing the coefficient of the highest power of $\partial$ in \eqref{KP:test eq}, we find $a_0=0$.
Then, step by step, one can successfully get
$a_{1}=a_2=\cdots=a_m=0$, which leads to $N=0$ and consequently $X=0$.
\end{proof}

The second lemma is
\begin{lem}\label{lem:KP-2}
The KP flows $\{K_{l}\}$ and $\{\sigma_{r}\}$ and operators $\{A_l\}$ and $\{B_r\}$ satisfy
\bse
\begin{align}
& 2\llbracket{K_{l},K_{r}}\rrbracket= {\langle}A_{l},A_{r}{\rangle}_y-[A_2,{\langle}A_{l},A_{r}\rangle] ,\\
& 2\llbracket{K_{l},\sigma_{r}}\rrbracket= {\langle}A_{l},B_{r}{\rangle}_y-[A_2,{\langle}A_{l},B_{r}\rangle] ,
\label{KP:Ks}\\
& 2\llbracket{\sigma_{l},\sigma_{r}}\rrbracket= {\langle}B_{l},B_{r}{\rangle}_y-[A_2,{\langle}B_{l},B_{r}\rangle] ,
\end{align}
\ese
where
\bse
\begin{align}
{\langle}A_{l},A_{r}{\rangle}&=A_{l}'[K_{r}]-A_{r}'[K_{l}]+[A_{l},A_{r}],\\
{\langle}A_{l},B_{r}{\rangle}&=A_{l}'[\sigma_{r}]-B_{r}'[K_{l}]+[A_{l},B_{r}],\label{KP:AB}\\
{\langle}B_{l},B_{r}{\rangle}&=B_{l}'[\sigma_{r}]-B_{r}'[\sigma_{l}]+[B_{l},B_{r}],
\end{align}
\ese
and satisfy
\bse
\begin{align}
{\langle}A_{l},A_{r}\rangle|_{u=0}&= 0,\\
{\langle}A_{l},B_{r}\rangle|_{u=0}&= l\,\partial^{l+r-2},
\label{KP:AB bc}\\
{\langle}B_{l},B_{r}\rangle|_{u=0}&= (l-r)\bigl(2y\partial^{l+r-2}+x\partial^{l+r-3}\bigr).
\end{align}
\ese
\end{lem}
\begin{proof}
We only prove \eqref{KP:Ks}. The others are similar.
From \eqref{KP:zcr} by direct calculation we find
\begin{align*}
2K_l'[\sigma_r]&=(A_l'[\sigma_r])_y-[2\sigma_r, A_l]-[A_2, A_l'[\sigma_r]]\\
               &=(A_l'[\sigma_r])_y-[B_{r,y}, A_l]+[[A_2,B_r],A_l]-[A_2, A_l'[\sigma_r]],
\end{align*}
and
\begin{align*}
2\sigma_r'[K_l]&=(B_r'[K_l])_y-[2 K_l, B_r]-[A_2, B_r'[K_l]]\\
               &=(B_r'[K_l])_y-[A_{l,y}, B_r]+[[A_2,A_l],B_r]-[A_2, B_r'[K_l]].
\end{align*}
Then, by substraction we reach to  \eqref{KP:Ks},
where we need to make use of the Jacobi identity
\[[[A,B],C]+[[B,C],A]+[[C,A],B]=0.\]
Besides, substituting the asymptotic data \eqref{KP:asym-a} and \eqref{KP:asym-b} into \eqref{KP:AB} we   get \eqref{KP:AB bc}.
We note that the method to prove this lemma has been used for many systems, e.g.
\cite{CZ-JPA-1991,Ma-JMP-1992,CZ-JMP-1996,MF-JMP-1999,TM-JPSJ-1999,Zhang-PLA-2006,ZC-SAM-2010-I}.
\end{proof}

These two lemmas together with the zero curvature representations \eqref{KP:zcr}
immediately lead to Theorem \ref{thm:KP alg}.

Theorem \ref{thm:KP alg} directly yields the following two corollaries.

\begin{cor}
\label{thm:iKP sym}
Each equation
\be
u_{t_{s}}=K_{s}
\ee
in the isospectral KP hierarchy \eqref{iKP:hie}
has two sets of symmetries
\be\label{iKP:sym}
\{K_{l}\},~~
\{\tau^{s}_{r}=st_sK_{s+r-2}+\sigma_r\}
\ee
and they generate a Lie algebra with basic structure
\bse\label{iKP:alg sym}
\begin{align}
& {\llbracket}K_{l},K_{r}{\rrbracket}= 0,\label{KP:sym-1}\\
& {\llbracket}K_{l},\tau^{s}_{r}{\rrbracket}= l\,K_{l+r-2},\label{KP:sym-2}\\
& {\llbracket}\tau^{s}_{l},\tau^{s}_{r}{\rrbracket}= (l-r)\tau^{s}_{l+r-2},\label{KP:sym-3}
\end{align}
\ese
where $l,r,s\geq 1$ and we set $K_0=\tau^{s}_{0}=0$.
\end{cor}

\begin{cor}
\label{thm:KP recur}
The master symmetry $\sigma_3$ acts as a flows generator via the following relation
\begin{subequations}
\label{KP:recur}
\begin{align}
&K_{s+1}=\frac{1}{s}\, \llbracket K_{s},\sigma_3 \rrbracket,
\label{iKP:recur}\\
&\sigma_{s+1}=\frac{1}{s-3}\,\llbracket \sigma_{s},\sigma_3 \rrbracket,~~(s\neq 3).
\label{nKP:recur}
\end{align}
\end{subequations}
with initial flows $K_1=u_x$ given in \eqref{iKP:flow} and $\sigma_1, \sigma_4$ given in \eqref{nKP:flow}.
\end{cor}

We note that $\sigma_3$ and the recursive relation \eqref{iKP:recur} are the same as those
given in \cite{OF-PLA-1982},
which means the KP hierarchy derived from Lax triad approach
and the KP hierarchy generated from the recursive structure in Ref.\cite{OF-PLA-1982} are same.

\subsection{Hamiltonian structures and conserved quantities}

In the literature\cite{OF-PLA-1982,Case-PNAS-1984,Case-JMP-1985,CLB-JPA-1988} it has been proved that both isospectral and non-isospectral KP hierarchies
have Hamiltonian structures and each equation in the isospectral KP hierarchy has two sets of conserved quantities.
We list these main results in the following two theorems.

\begin{thm}
\label{thm:iKP Ham}
Each equation in the isospectral KP hierarchy has a Hamiltonian structure, i.e.
\be
u_{t_s}=K_s=\partial \frac{\delta H_s}{\delta u},
\label{iKP:Ham}
\ee
where the gradient functions $\gamma_s=\frac{\delta H_s}{\delta u}$  is defined by
\be
\gamma_s=\partial^{-1}K_s=
\left\{\begin{array}{ll}
u,~~& s=1,\\
\frac{1}{s-1}\,{\rm grad}(\gamma_{s-1},\sigma_3),~& s>1.
\end{array}\right.
\ee
The equation \eqref{iKP:Ham} has infinitely many conserved quantities
\be
H_1=\frac{1}{2}(u,u),~~H_s=\frac{1}{s-1}(\gamma_{s-1},\sigma_3),~~(s>1).
\ee
\end{thm}

The key identity that leads to the above theorem is
\be
\sigma_3'\partial+\partial\sigma_3'^*=0.
\label{KP-to ham}
\ee
That means $\partial$ is a Noether operator of the non-isospectral equation $u_{t_3}=\sigma_3$.

\begin{thm}\label{T:3-3}
(1). Each equation
\be
u_{t_s}=\sigma_s
\ee
in the non-isospectral KP hierarchy \eqref{nKP:hie}
has a Hamiltonian structure
\be
u_{t_s}=\sigma_s=\partial \frac{\delta J_s}{\delta u},
\label{nKP:Ham}
\ee
where the gradient function
\begin{align}
\omega_s&=\frac{\delta J_s}{\delta u}=\partial^{-1}\sigma_s\nonumber\\
&=
\left\{\begin{array}{ll}
yu,~~& s=1,\\
\frac{1}{s-4}\,{\rm grad}(\omega_{s-1},\sigma_3),~& s>1,s\neq 4,\\
2y\gamma_4+x\gamma_3+\frac{3}{4}u_x+\frac{3}{2}\partial^{-1}u^2+\frac{3}{4}\partial^{-3}u_{yy}+u\partial^{-1}u-\frac{3}{2}\partial^{-1}u_y,~~& s=4.\\
\end{array}\right.
\end{align}
The Hamiltonian is
\be
J_s=\left\{\begin{array}{ll}
\frac{1}{2}(y u,u),~~& s=1,\\
\frac{1}{s-4}(\omega_{s-1},\sigma_3),~& s>1, s\neq 4,\\
\int_0^1(\omega_4(\lambda u),u)d\lambda,~~& s=4.
\end{array}\right.
\ee

\noindent
(2). Hamiltonians $\{H_l\}$ and $\{J_r\}$ generate a Lie algebra  w.r.t. Poisson bracket $\{\cdot, \cdot\}$ with basic structure
\bse
\label{iKP:alg cc}
\begin{align}
&\{ H_l, H_r\}=0,\label{iKP:cc-1}\\
&\{ H_l, J_r\}=l \, H_{l+r-2},\label{iKP:cc-2}\\
&\{ J_l, J_r\}=(l-r) J_{l+r-2},\label{iKP:cc-3}
\end{align}
\ese
where $l,r,s\geq 1$ and we set $H_0=J_{0}=0$. Here the  Poisson bracket is defined as
\[\{H,J\}=\Bigl(\frac{\delta H}{\delta u},\partial\frac{\delta J}{\delta u}\Bigr).\]

\noindent
(3). Each equation
\be
u_{t_{s}}=K_{s}
\ee
in the isospectral KP hierarchy \eqref{iKP:hie}
has two sets of conserved quantities
\begin{align}\label{iKP:cq}
\{H_{l}\},~~
\{I^{s}_{r}=st_s H_{s+r-2}+ J_r\}
\end{align}
and they generate a Lie algebra with basic structure
\bse\label{iKP:alg cq}
\begin{align}
& \{H_{l},H_{r}\}= 0,\label{KP:cq-1}\\
& \{H_{l},I_{r}^{s}\}= l\,H_{l+r-2},\label{KP:cq-2}\\
& \{I^{s}_{l},I_{r}^{s}\}= (l-r)I^{s}_{l+r-2},\label{KP:cq-3}
\end{align}
\ese
where $l,r,s\geq 1$ and we set $H_0=I^{s}_{0}=0$.
\end{thm}

\section{The D$\Delta$KP system}
\label{S:4}
In this section, we will construct the  D$\Delta$KP hierarchy and discuss their recursive structure,
symmetries, Hamiltonian structures and conserved quantities.

\subsection{The D$\Delta$KP hierarchy}\label{S:4-1}

Let us start from the following linear triad
\bse
\begin{align}
& \b L\phi=\eta\phi,~~\eta_{\b t_m}=0,\\
& \phi_{\b x}=\b A_1\phi,~~\b A_{1}=h^{-1}\Delta+\b u_{0},\\
& \phi_{\b t_{m}}=\b A_{m}\phi,~~ (m=1,2,\cdots),
\end{align}
\ese
and the compatibility condition reads
\bse\label{diKP:com}
\begin{align}
&\b L_{\b x}=[\b A_{1},\b L],\label{diKP:com-1} \\
&\b L_{\b t_{m}}=[\b A_{m},\b L],\label{diKP:com-2} \\
&\b A_{1,\b t_m}-\b A_{m,\b x}+[\b A_1,\b A_m]=0,\label{diKP:com-3}
\end{align}
\ese
for $m=1,2,\cdots$,
where $\b L$ is the pseudo-difference operator \eqref{dKP:L}, $\b A_{m}=(\b L^{m})_{+}$ with the form
\be\label{diKP:A bc}
\b A_{m}=h^{-m}\Delta^{m}+\sum^{m}_{j=1}h^{-(m-j)}\b a_j\Delta^{m-j},~~~\b A_{s}|_{\b{\mathbf{u}}=\mathbf{0}}=h^{-m}\Delta^{m}.
\ee
Here $\b{\mathbf{u}}=(\b u_0, \b u_1, \cdots)$.
The first three of $\b A_{m}$ are
\bse\label{diKP:A}
\begin{align}
\b A_{1}&=h^{-1}\Delta+\b u_{0},\label{diKP:A-1}\\
\b A_{2}&=h^{-2}\Delta^{2}+h^{-1}({\Delta}\b u_{0}+2\b u_{0})\Delta+(h^{-1}{\Delta}\b u_{0}+\b u_{0}^{2}+{\Delta}\b u_{1}+2\b u_{1}),\label{diKP:A-2} \\
\b A_{3}&=h^{-3}\Delta^{3}+h^{-2}\b a_{1}\Delta^{2}+h^{-1}\b a_{2}\Delta+\b a_{3},\label{diKP:A-3}
\end{align}
\ese
where
\bse
\begin{align*}
\b a_{1}&=\Delta^{2}\b u_{0}+3{\Delta}\b u_{0}+3\b u_{0},\\
\b a_{2}&=2h^{-1}\Delta^{2}\b u_{0}+3h^{-1}{\Delta}\b u_{0}+3\b u_{0}^{2}+\b u_{0}{\Delta}\b u_{0}
+{\Delta}\b u_{0}^{2}+3\b u_{1}+3{\Delta}\b u_{1}+\Delta^{2}\b u_{1},\\
\b a_{3}&=h^{-2}\Delta^{2}\b u_{0}+\b u_{0}^{3}+h^{-1}\b u_{0}{\Delta}\b u_{0}+h^{-1}{\Delta}\b u_{0}^2
+5\b u_{0}\b u_{1}+({\Delta}\b u_{0}){\Delta}\b u_{1}+3\b u_{0}{\Delta}\b u_{1} \nonumber\\
&~~~+\b u_{1}{\Delta}\b u_{0}+\b u_{1}E^{-1}\b u_{0}+2h^{-1}\Delta^{2}\b u_{1}+3h^{-1}{\Delta}\b u_{1}
+3\b u_{2}+3{\Delta}\b u_{2}+\Delta^{2}\b u_{2}.
\end{align*}
\ese
Equation \eqref{diKP:com-1} yields
\bse
\begin{align}
\b u_{0,\b x}&={\Delta}\b u_{1},\\
\b u_{1,\b x}&=h^{-1}{\Delta}\b u_{1}+{\Delta}\b u_{2}+\b u_{0}\b u_{1}-\b u_{1}E^{-1}\b u_{0},\\
&~~\cdots\cdots,\nonumber
\end{align}
\ese
which will be used to express $\b u_j (j>0)$ by $\b u_0$,
i.e.
\bse\label{dKP:uj}
\begin{align}
\b u_{1}&=\Delta^{-1}\frac{\partial{\b u_{0}}}{\partial\b x},\label{dKP:u1}\\
\b u_{2}&=\Delta^{-2}\frac{\partial^{2}{\b u_{0}}}{\partial{\b x^{2}}}-h^{-1}\Delta^{-1}\frac{\partial{\b u_{0}}}{\partial{\b x}}
-\Delta^{-1}\biggl(\b u_0\Delta^{-1}\frac{\partial{\b u_0}}{\partial\b x}\biggr)
+\Delta^{-1}\biggl(\biggl(\Delta^{-1}\frac{\partial{\b u_0}}{\partial\b x}\biggr)E^{-1}\b u_0\biggr),\label{dKP:u2}\\
&~~\cdots\cdots.\nonumber
\end{align}
\ese
Equation \eqref{diKP:com-2} actually plays a role to determine the operator $\b A_m$.
In fact, if starting from the assumption \eqref{diKP:A bc} with unknown $\{\b a_j\}$, then \eqref{diKP:com-2}
uniquely determines $\b A_m=(\b L^{m})_{+}$.
With $\{\b A_m\}$ ready, equation \eqref{diKP:com-3} provides the isospectral D$\Delta$KP hierarchy (with $\b u_0=\b u$)
\be\label{diKP:hie}
\b u_{\b t_m}=\b K_m=\b A_{m,\b x}-[\b A_1,\b A_m],~~m=1,2, \cdots.
\ee
The first three equations are
\bse\label{diKP:flow}
\begin{align}
\b u_{\b t_{1}}=\b K_{1}&=\b u_{\b x},\label{diKP:flow-1}\\
\b u_{\b t_{2}}=\b K_{2}&=(1+2\Delta^{-1})\b u_{\b x\b x}-2h^{-1}\b u_{\b x}+2\b u \b u_{\b x},\label{diKP:flow-2}\\
\b u_{\b t_{3}}=\b K_{3}&=(3\Delta^{-2}+3\Delta^{-1}+1)\b u_{\b x\b x\b x}+3\Delta^{-1}\b u_{\b x}^2+3\b u\Delta^{-1}\b u_{\b x\b x} \nonumber \\
&~~~-6 h^{-1}\Delta^{-1}\b u_{\b x\b x} +3h^{-2}\b u_{\b x}+3\b u_{\b x}\Delta^{-1}\b u_{\b x}+3\Delta^{-1}(\b u\b u_{\b x\b x}) \nonumber \\
&~~~+3\b u\b u_{\b x\b x}-3h^{-1}\b u_{\b x\b x}+3\b u_{\b x}^2+3\b u^2\b u_{\b x}-6h^{-1}\b u\b u_{\b x},\label{diKP:flow-3}
\end{align}
\ese
in which \eqref{diKP:flow-2}, i.e. \eqref{DDKP}, is first derived in \cite{DJM-JPSJ-1982-II}
from a discrete Sato's approach and is referred to as the D$\Delta$KP equation.

\subsection{The non-isospectral D$\Delta$KP hierarchy}\label{S:4.2}

For the
non-isospectral case, we set
\be\label{dnKP:eta}
\eta_{\b t_{m}}=h\eta^{m}+\eta^{m-1},
\ee
and assume that
\be
\b B_{m}=\sum^{m}_{j=0}h^{-(m-j)}\b b_{j}\Delta^{m-j}
\ee
with unknowns $\{\b b_{j}\}$. Consider the Lax triad
\bse\label{dnKP:triad}
\begin{align}
& \b L\phi =\eta\phi, \label{dnKP:triad-1}\\
& \phi_{\b x} =\b A_1\phi, \label{dnKP:triad-2}\\
& \phi_{\b t_{m}} =\b B_{m}\phi,~~ (m=1,2,\cdots), \label{dnKP:triad-3}
\end{align}
\ese
together with \eqref{dnKP:eta}.
The compatibility reads
\bse\label{dnKP:com}
\begin{align}
& \b L_{\b x}=[\b A_{1},\b L],\label{dnKP:com-1}\\
& \b L_{\b t_{m}}=[\b B_{m},\b L]+h\b L^{m}+\b L^{m-1},\label{dnKP:com-2}\\
& \b A_{1,\b t_m}-\b B_{m,\b x}+[\b A_1,\b B_m]=0.\label{dnKP:com-3}
\end{align}
\ese
Looking at \eqref{dnKP:com-2} and \eqref{dnKP:com-3} asymptotically, i.e. \eqref{dnKP:com-2}$|_{\mathbf{\b u}=\mathbf{0}}$
and \eqref{dnKP:com-3}$|_{\mathbf{\b u}=\mathbf{0}}$, from them  one can find
\begin{align*}
& (\Delta\b b_0)|_{\mathbf{\b u}=\mathbf{0}}=0,~~(\Delta\b b_1)|_{\mathbf{\b u}=\mathbf{0}}=h,~~(\Delta\b b_j)|_{\mathbf{\b u}=\mathbf{0}}=0,
~(j=2,3,\cdots,m);\\
&(\partial_{\b x}\b b_0)|_{\mathbf{\b u}=\mathbf{0}}=h,~~(\partial_{\b x}\b b_1)|_{\mathbf{\b u}=\mathbf{0}}=1,
~~(\partial_{\b x}\b b_j)|_{\mathbf{\b u}=\mathbf{0}}=0,~(j=2,3,\cdots,m).
\end{align*}
This gives the necessary asymptotic condition for $\b B_m$:\footnote{ In \cite{SZZC-MPLB-2010}
the asymptotic condition for $B_m$ is
$B_{m}|_{\b{\mathbf{u}}=\mathbf{0}}=h^{-(m-1)}\b x\Delta^{m}+h^{-(m-2)}n\Delta^{m-1}$.
We note that this is not sufficient due to missing \eqref{dnKP:triad-2} in the Lax triad \eqref{dnKP:triad}.
}
\be
\b B_{m}|_{\b{\mathbf{u}}=\mathbf{0}}=h^{-(m-1)}\b x\Delta^{m}+h^{-(m-1)}(\b x +h n)\Delta^{m-1}.
\label{nKP:B bc}
\ee
Then, with this condition, $\b B_m$ can uniquely be determined by \eqref{dnKP:com-2} and here we give the first three of them:
\bse
\begin{align}
\b B_{1}&=h\b x\b A_{1}+\b x+hn,\\
\b B_{2}&=h\b x\b A_{2}+(\b x+hn)\b A_{1}+h\Delta^{-1}\b u_{0},\\
\b B_{3}&=h\b x\b A_{3}+(\b x+hn)\b A_{2}+\Delta^{-1}\b u_{0}\Delta+h\b u_{0}\Delta^{-1}\b u_{0}\nonumber  \\
&~~~+2h\Delta^{-1}\b u_{1}-\Delta^{-1}\b u_{0}+h\Delta^{-1}\b u_{0}^{2},
\end{align}
\ese
where $\b A_j=(\b L^j)_+$.

Now, \eqref{dnKP:com-1} provides transform relation as same as
\eqref{dKP:uj}, and \eqref{dnKP:com-3} provides the non-isospectral D$\Delta$KP hierarchy (with $\b u_{0}=\b u$)
\be\label{dnKP:hie}
\b u_{t_m}=\b\sigma_m=\b B_{m,\b x}-[\b A_1,\b B_m],
\ee
i.e.
\bse
\label{dnKP:flow}
\begin{align}
\b u_{\b t_{1}}=\b \sigma_{1}&=h\b x\b K_{1}+h\b u,
\label{dnKP:flow-1}\\
\b u_{\b t_{2}}=\b \sigma_{2}&=h\b x\b K_{2}+(\b x+hn)\b K_1+h\b u_{\b x}+3h\Delta^{-1}\b u_{\b x}+h\b u^{2}-\b u,
\label{dnKP:flow-2}\\
\b u_{\b t_{3}}=\b \sigma_{3}&=h\b x\b K_{3}+(\b x+hn)\b K_2+5h\Delta^{-2}\b u_{\b x\b x}
-6\Delta^{-1}\b u_{\b x}+5h\Delta^{-1}(\b u\b u_{\b x}) \nonumber \\
&~~~+h\b u_{\b x}\Delta^{-1}\b u+4h\b u\Delta^{-1}\b u_{\b x}-2\b u^2+h^{-1}\b u+h\b u^3+3h\b u\b u_{\b x} \nonumber \\
&~~~+3h\Delta^{-1}\b u_{\b x\b x}+h\b u_{\b x\b x}-2\b u_{\b x},
\label{dnKP:flow-3}\\
&~~\cdots \cdots,\nonumber
\end{align}
\ese
where $\{\b K_{j}\}$ are the isospectral D$\Delta$KP flows defined in \eqref{diKP:hie}.

$\{\b K_m\}$ and $\{\b{\sigma}_m\}$ are respectively called the isospectral and non-isospectral D$\Delta$KP flows.
For them we have
\begin{prop}\label{P:4-1}
The isospectral and non-isospectral D$\Delta$KP flows
$\{\b K_s\}$ and $\{\b{\sigma}_s\}$
can be expressed through the following zero curvature representations together with asymptotic conditions,
\bse\label{dKP:zcr}
\begin{align}
&\b K_s =\b A_{s,\b x}-[\b A_1, \b A_s],~~~\b K_s|_{\b u=0}=0,~~ \b A_{s}|_{\mathbf{\b u}=\mathbf{0}}=h^{-s}\Delta^{s},\label{diKP:zcr}\\
& \b \sigma_s =\b B_{s,\b x}-[\b A_1, \b B_s],~~~\b \sigma_s|_{\b u=0}=0,~~\b
B_{s}|_{{\mathbf{\b u}}=\mathbf{0}}=h^{-(s-1)}\b x\Delta^{s}+h^{-(s-1)}(\b x +h n)\Delta^{s-1},\label{dnKP:zcr}
\end{align}
\ese
for $s=1,2,\cdots$.
\end{prop}

Similar to Proposition \ref{P:3-2}, we have
\begin{prop}\label{P:4-2}
The isospectral D$\Delta$KP flows $\{\b K_s\}$ defined by \eqref{diKP:zcr} can be expressed in terms of the pseudo-difference operator $\b L$ as
\be
\b K_s=\Delta \,\underset{\Delta}{\mathrm{Res\,}}\b L^s.
\label{bK-Res}
\ee
\end{prop}
Proof is skipped.

\subsection{Algebra of flows, recursive structure and symmetries }
\label{S:4.3}

The proof for the results of this subsection is similar to the continuous  case (see Sec.\ref{S:3.4}).
We will just list these results without giving proofs.

\begin{lem}\label{lem:dKP-1}
For the function $\b X=\b X(\b u)\in \b{\mathcal{F}}$ and difference operator
\begin{align*}
\b N=\b a_0\Delta^m+\b a_1\Delta^{m-1}+\cdots+\b a_{m-1}\Delta+\b a_m,~~\b N|_{\b u=0}=0
\end{align*}
living on $\b{\mathcal{F}}$,
the equation
\be\label{dKP:test eq}
\b X-\b N_{\b x}+[\b A_1, \b N]=0
\ee
only admits zero solution $\b X=0,~\b N=0$.
\end{lem}

\begin{lem} \label{lem:dKP-2}
Suppose that
\bse
\begin{align}
{\langle}\b A_{l},\b A_{r}{\rangle}&=\b A_{l}'[\b K_{r}]-\b A_{r}'[\b K_{l}]+[\b A_{l},\b A_{r}],\\
{\langle}\bar A_{l},\b B_{r}{\rangle}&=\b A_{l}'[\b \sigma_{r}]-\b B_{r}'[\b K_{l}]+[\b A_{l},\b B_{r}],\\
{\langle}\b B_{l},\b B_{r}{\rangle}&=\b B_{l}'[\b \sigma_{r}]-\b B_{r}'[\b \sigma_{l}]+[\b B_{l},\b B_{r}].
\end{align}
\ese
Then we have
\bse
\begin{align}
 & \llbracket{\b K_{l},\b K_{r}}\rrbracket={\langle}\b A_{l},\b A_{r}\rangle_{\b x}-[\b A_1,{\langle}\b A_{l},\b A_{r}\rangle],\\
 & \llbracket{\b K_{l},\b \sigma_{r}}\rrbracket={\langle}\b A_{l},\b B_{r}\rangle_{\b x}-[\b A_1,{\langle}\b A_{l},\b B_{r}\rangle],\\
 & \llbracket{\b \sigma_{l},\b \sigma_{r}}\rrbracket={\langle}\b B_{l},\b B_{r}\rangle_{\b x}- [\b A_1,{\langle}\b B_{l},\b B_{r}\rangle],
\end{align}
\ese
and
\bse
\begin{align}
&{\langle}\b A_{l},\b A_{r}\rangle|_{\b u=0}=0,\\
&{\langle}\b A_{l},\b B_{r}\rangle|_{\b u=0}=h^{-(l+r-2)}\,l\,(\Delta^{l+r-1}+\Delta^{l+r-2}),\\
&{\langle}\b B_{l},\b B_{r}\rangle|_{\b u=0}=h^{-(l+r-3)}(l-r)\bigl(\b x \Delta^{l+r-1}+(2\b x+hn)\Delta^{l+r-2}+(\b x+hn)\Delta^{l+r-3}\bigr).
\end{align}
\ese
\end{lem}

\begin{thm}\label{thm:dKP alg}
The flows
$\{\b K_{l}\}$ and $\{\b \sigma_{r}\}$  span a Lie algebra $\b {\mathbf{X}}$ with basic structure
\bse\label{dKP:alg}
\begin{align}
&{\llbracket}\b K_{l},\b K_{r}{\rrbracket}=0,\label{dKP:alg-1}\\
&{\llbracket}\b K_{l},\b \sigma_{r}{\rrbracket}=l\,(h\b K_{l+r-1}+\b K_{l+r-2}),\label{dKP:alg-2}\\
&{\llbracket}\b \sigma_{l},\b \sigma_{r}{\rrbracket}=(l-r)(h\b\sigma_{l+r-1}+\b \sigma_{l+r-2}),\label{dKP:alg-3}
\end{align}
\ese
where $l,r\geq 1$ and we set $\b K_0=\b \sigma_0=0$.
\end{thm}

\begin{cor}\label{thm:diKP sym}
Each equation
\be
\b u_{t_{s}}=\b K_{s}
\ee
in the isospectral D$\Delta$KP hierarchy \eqref{diKP:hie}
possesses two sets of symmetries
\be\label{diKP:sym}
 \{\b K_{l}\},~~
 \{\b\tau^{s}_{r}=s\b t_{s}(h\b K_{s+r-1}+\b K_{s+r-2})+\b \sigma_{r}\},
\ee
which generate a Lie algebra with basic structure
\bse\label{diKP:alg sym}
\begin{align}
& {\llbracket}\b K_{l},\b K_{r}{\rrbracket}=0,\\
& {\llbracket}\b K_{l},\b \tau^{s}_{r}{\rrbracket}=l\,(h\b K_{l+r-1}+\b K_{l+r-2}),\\
& {\llbracket}\b \tau^{s}_{l},\b \tau^{s}_{r}{\rrbracket}=(l-r)(h\b\tau^{s}_{l+r-1}+\b\tau^{s}_{l+r-2}),\label{diKP:alg sym-c}
\end{align}
\ese
where $l,r,s\geq 1$ and we set $\b K_0=\b \tau^{s}_{0}=0$.
\end{cor}

\begin{cor}
\label{cor:dKP recur}
$\b \sigma_2$ is a  master symmetry and acts as a flows generator via the following relation
\begin{subequations}
\begin{align}
&\b K_{s+1}=\frac{1}{h}\biggl(\frac{1}{s}\llbracket\b K_{s},\b\sigma_2\rrbracket-\b K_{s}\biggr),
\label{diKP:recur}\\
&\b \sigma_{s+1}=\frac{1}{h}\biggl(\frac{1}{s-2}\llbracket\b \sigma_{s},\b\sigma_2\rrbracket-\b\sigma_s\biggr),~~(s\neq 2),
\label{dnKP:recur}
\end{align}
\end{subequations}
with initial flows $\b K_1=\b u_{\b x}$ given in \eqref{diKP:flow} and $\b \sigma_1, \b \sigma_3$ given in \eqref{dnKP:flow}.
\end{cor}

\subsection{Hamiltonian structures and conserved quantities}
\label{S:4.4}

For (1+1)-dimensional Lax integrable systems, they usually have their own recursion operators
that play crucial roles in investigating integrable characteristics (cf. \cite{FF-PD-1981,ZC-JPA-2002,ZC-SAM-2010-I,ZC-SAM-2010-II}).
For  the isospectral D$\Delta$KP hierarchy, so far there is no explicit recursion operator
but their recursive structure \eqref{diKP:recur} will play a similar role.
We will show that
each member in the isospectral D$\Delta$KP hierarchy \eqref{diKP:hie} and non-isospectral D$\Delta$KP hierarchy \eqref{dnKP:hie}
has a Hamiltonian structure, and the Hamiltonians lead to two sets of  conserved quantities
for the isospectral D$\Delta$KP hierarchy \eqref{diKP:hie}.
Let us prove this step by step.

\begin{lem}
\label{lem:grad}
The following formula
\be
\mathrm{grad}(\b\gamma,\b\sigma)=\b\gamma'^*\b\sigma+\b\sigma'^*\b\gamma
\ee
holds for any $\b\gamma,\b\sigma\in\b{\mathcal{F}}$.
\end{lem}
In fact, one can verify that
\begin{align*}
(\b\gamma,\b\sigma)'[\b g]=(\b\gamma'[\b g],\b\sigma)+(\b\gamma,\b\sigma'[\b g])
=(\b\gamma'^*\b\sigma+\b\sigma'^*\b\gamma,\b g),~~\forall \b g=\b g(\b u)\in \b{\mathcal{F}}.
\end{align*}
\begin{lem}\label{lem:master eq}
$\partial_{\b x}$ and $\b\sigma_2$ satisfy
\be
\b\sigma_2'\partial_{\b x}+\partial_{\b x}\b\sigma_2'^*=0,
\label{sig-par}
\ee
i.e. $\partial_{\b x}$ is a Noether operator of the master symmetry equation $\b u_{\b t_2}=\b \sigma_2$.
\end{lem}
This identity is important for getting Hamiltonian structures for both isospectral and non-isospectral D$\Delta$KP hierarchies.
The proof of \eqref{sig-par} will lead to lengthy but direct calculation and here we skip it.

Now we arrive at the first main theorem of this subsection.

\begin{thm}
\label{thm:diKP Ham}
Each equation
\be
\b u_{t_{s}}=\b K_{s}
\label{diKP:eq}
\ee
in the isospectral D$\Delta$KP hierarchy \eqref{diKP:hie} has a Hamiltonian structure
\be
\b u_{\b t_s}=\b K_s=\partial_{\b x}\frac{\delta \b H_s}{\delta \b u},
\label{diKP:Ham}
\ee
where
\be
\frac{\delta \b H_s}{\delta \b u}=\b \gamma_s=\partial_{\b x}^{-1} \b K_s.
\ee
$\b \gamma_s$ can also be determined through
\be
\b \gamma_{s}=
\left\{
\begin{array}{ll}
\b u, & ~s=1,\\
\frac{1}{h}\bigl(\frac{1}{s-1}\,{\rm grad}(\b\gamma_{s-1},\b\sigma_2)-\b\gamma_{s-1}\bigr), & ~s>1.
\end{array}
\right.
\label{diKP:gamma recur}
\ee
The Hamiltonian $\b H_s$ can be given by
\be
\b H_{s}=
\left\{
\begin{array}{ll}
\frac{1}{2}(\b u, \b u), & ~s=1,\\
\frac{1}{h}\bigl(\frac{1}{s-1}(\b\gamma_{s-1},\b\sigma_2)-\b H_{s-1}\bigr), & ~s>1.

\end{array}
\right.
\label{diKP:H recur}
\ee
\end{thm}

\begin{proof}
Obviously, $\partial_{\b x}$ is an implectic operator.
Next we need to prove $\b \gamma_s$ is a gradient function.
Let us do that by means of mathematical inductive method.
Obviously, $\gamma_1=\b u$ is a gradient function.
We suppose $\b\gamma_s$ is a gradient function, i.e. $\b\gamma_s'=\b\gamma_s'^*$.
Then, from the recursive relation \eqref{diKP:recur} we have
\begin{align*}
\b\gamma_{s+1}&=\partial_{\b x}^{-1} \b K_{s+1} \\
&=\frac{1}{h}\biggl(\frac{1}{s}\,\partial_{\b x}^{-1}{\llbracket}\b K_{s},\b\sigma_{2}{\rrbracket}-\partial_{\b x}^{-1}\b K_s\biggr) \\
&=\frac{1}{h}\biggl(\frac{1}{s}\,\partial_{\b x}^{-1}\big((\partial_{\b x}\b\gamma_s)'[\b\sigma_2]-\b\sigma_2'[\partial_{\b x}\b\gamma_s]\big)
-\b\gamma_s\biggr)\\
&=\frac{1}{h}\biggl(\frac{1}{s}\,\partial_{\b x}^{-1}(\partial_{\b x}\b\gamma_s' \b\sigma_2 -\b\sigma_2'\partial_{\b x}\,\b\gamma_s)-\b\gamma_s\biggr).
\end{align*}
It then follows from Lemma \ref{lem:master eq} that
\begin{align*}
\b\gamma_{s+1}&=\frac{1}{h}\biggl(\frac{1}{s}\,\partial_{\b x}^{-1}(\partial_{\b x}\,\b\gamma_s' \b\sigma_2
+ \partial_{\b x}\b\sigma_2'^*\,\b\gamma_s)-\b\gamma_s\biggr)\\
&=\frac{1}{h}\biggl(\frac{1}{s}\,(\b\gamma_s' \b\sigma_2 + \b\sigma_2'^*\,\b\gamma_s)-\b\gamma_s\biggr)\\
&=\frac{1}{h}\biggl(\frac{1}{s}\,{\rm grad}(\b\gamma_{s},\b\sigma_2)-\b\gamma_s\biggr).
\end{align*}
Here we also made use of Lemma \ref{lem:grad}.
This means if $\b \gamma_s$ is a gradient function, so is $\b\gamma_{s+1}$.
For the Hamiltonians, $\b H_1=\frac{1}{2}(\b u, \b u)$
is derived from $\b \gamma_1=\b u$ and Proposition \ref{P:2-1}.
$\b H_s$ for $s>1$ follows from the recursive relation of $\b \gamma_s$ given in \eqref{diKP:gamma recur}.
We complete the proof.
\end{proof}

The non-isospectral D$\Delta$KP hierarchy \eqref{dnKP:hie} also have their own Hamiltonian structures.

\begin{thm}
\label{thm:dnKP Ham}
Each equation
\be
\b u_{t_{s}}=\b \sigma_{s}
\label{dnKP:eq}
\ee
in the non-isospectral D$\Delta$KP hierarchy \eqref{dnKP:hie} has a Hamiltonian structure
\be
\b u_{\b t}=\b \sigma_s=\partial_{\b x}\frac{\delta \b J_s}{\delta \b u},
\label{dnKP:Ham}
\ee
where
\be
\frac{\delta \b J_s}{\delta \b u}=\b \omega_s=\partial_{\b x}^{-1} \b \sigma_s.
\ee
and
\begin{subequations}
\begin{align}
& \b \omega_1=h\b x\b u, \\
& \b \omega_2=h\b x\b\gamma_2+(\b x+hn)\b\gamma_1+h\Delta^{-1}\b u,\\
& \b \omega_3=h\b x\b\gamma_3+(\b x+hn)\b\gamma_2+2h\Delta^{-2}\b u_{\b x}+h\Delta^{-1}\b u^2+h\b u\Delta^{-1}\b u-2\Delta^{-1}\b u, \\
& \b \omega_{s}=\frac{1}{h}\biggl(\frac{1}{s-3}\,{\rm grad}(\b\omega_{s-1},\b\sigma_2)-\b\omega_{s-1}\biggr),~~s=4,5,\cdots .
\label{diKP:omega recur}
\end{align}
\end{subequations}
The Hamiltonian $\b J_s$ can be given by
\be
\b J_{s}=
\left\{
\begin{array}{ll}
\frac{h}{2}(\b x\b u, \b u), & ~s=1,\\
\frac{1}{h}\bigl(\frac{1}{s-3}(\b\omega_{s-1},\b\sigma_2)-\b J_{s-1}\bigr), & ~s>1,~s\neq 3,\\
\int^{1}_{0}(\b \omega_3(\lambda \b u), \b u)d \lambda, & ~s=3.
\end{array}
\right.
\label{diKP:J recur}
\ee
\end{thm}

\begin{proof}
The proof is quite similar to the previous theorem. Here we need to start from the recursive relation \eqref{dnKP:recur}.
We note that $\b\sigma_3$ can not be derived from \eqref{dnKP:recur}.
By direct verification we can find $\b\omega_s'=\b\omega_s'^*$ holds for $s=1,2,3$.
Now we suppose that $\b\omega_s$ is a gradient function.
Then, if $s>2$, from the recursive relation \eqref{dnKP:recur} we have
\begin{align*}
\b\omega_{s+1}&=\partial_{\b x}^{-1} \b \sigma_{s+1}
=\frac{1}{h}\biggl(\frac{1}{s-2}\,\partial_{\b x}^{-1}{\llbracket}\b\sigma_{s},\b\sigma_{2}{\rrbracket}-\partial_{\b x}^{-1}\b \sigma_s\biggr) \\
&=\frac{1}{h}\biggl(\frac{1}{s-2}\,\partial_{\b x}^{-1}(\partial_{\b x}\,\b\omega_s' \b\sigma_2 -\b\sigma_2'\partial_{\b x}\,\b\omega_s)-\b\omega_s\biggr)\\
&=\frac{1}{h}\biggl(\frac{1}{s-2}\,\partial_{\b x}^{-1}(\partial_{\b x}\,\b\omega_s' \b\sigma_2 + \partial_{\b x}\b\sigma_2'^*\,\b\omega_s)-\b\omega_s\biggr)\\
&=\frac{1}{h}\biggl(\frac{1}{s-2}\,(\b\omega_s' \b\sigma_2 + \b\sigma_2'^*\,\b\omega_s)-\b\omega_s\biggr)\\
&=\frac{1}{h}\biggl(\frac{1}{s-2}\,{\rm grad}(\b\omega_{s},\b\sigma_2)-\b\omega_s\biggr),
\end{align*}
where we have made use of Lemma \ref{lem:master eq} and Lemma \ref{lem:grad}.
Since $\b \omega_s$ is a gradient function, so is $\b\omega_{s+1}$,
and therefore \eqref{dnKP:Ham} holds.
The Hamiltonian $\b J_s$ is defined following Proposition \ref{P:2-1} and Lemma \ref{lem:grad}.
\end{proof}

\begin{cor}\label{C:4-3}
$\partial_{\b x}$ is a Noether operator for both isospectral   D$\Delta$KP hierarchy \eqref{diKP:hie}
and non-isospectral D$\Delta$KP hierarchy \eqref{dnKP:hie}.
\end{cor}
\begin{proof}
Consider an arbitrary isospectral equation
\be
\b u_{\b t_s}=\b K_s
\ee
in \eqref{diKP:hie}.
We only need to prove
\be
\b K_s' \partial_{\b x}+\partial_{\b x} \b K_s'^*=0.
\label{4.42}
\ee
In fact,
\begin{align*}
\b K_s' \partial_{\b x}+\partial_{\b x} \b K_s'^*
=(\partial_{\b x}\b\gamma_s)'\partial_{\b x}+\partial_{\b x}(\partial_{\b x}\b\gamma_s)'^*
=\partial_{\b x}\b\gamma_s'\partial_{\b x}-\partial_{\b x}\b\gamma_s'^*\partial_{\b x},
\end{align*}
which is zero due to $\b\gamma_s$ being a gradient function,
i.e. $\b\gamma_s'=\b\gamma_s'^*$.
In a same way and noting that $\b\omega_s'=\b\omega_s'^*$, we can prove that
\be
\b \sigma_s' \partial_{\b x}+\partial_{\b x} \b \sigma_s'^*=0,
\label{eq-sg-m}
\ee
which means $\partial_{\b x}$ is also a Noether operator of the  isospectral   equation
$\b u_{\b t_s}=\b \sigma_s$.
\end{proof}

Now we reach to the final theorem of this section.

\begin{thm}\label{T:4-4}
Each equation
\be
\b u_{\b t_{s}}=\b K_{s}
\label{eq-s}
\ee
in the isospectral D$\Delta$KP hierarchy \eqref{diKP:hie}
has two sets of conserved quantities
\be\label{diKP:cq}
\{\b H_{l}\},~~
\{\b I^{s}_{r}=s\b t_s (h\b H_{s+r-1}+\b H_{s+r-2})+ \b J_r\},
\ee
where $\b H_l$ and $\b J_r$ are defined in \eqref{diKP:H recur} and \eqref{diKP:J recur}, respectively.
They generate a Lie algebra w.r.t. Poisson bracket $\{\cdot, \cdot\}$ with basic structure
\bse\label{diKP:alg cq}
\begin{align}
& \{\b H_{l},\b H_{r}\}= 0,\label{diKP:cq-1}\\
&\{\b H_l,\b I_r^s\}=l\,(h\b H_{l+r-1}+\b H_{l+r-2}),\label{diKP:cq-2}\\
&\{\b I_l^s,\b I_r^s\}=(l-r)(h\b I_{l+r-1}^s+\b I_{l+r-2}^s),\label{diKP:cq-3}
\end{align}
\ese
where $l,r,s\geq 1$ and we set $\b H_0=\b I^{s}_{0}=0$.
\end{thm}

\begin{proof}
First, let us prove that both $\b H_l$ and $\b I^s_r$ are conserved quantities of equation \eqref{eq-s}.
Noting that
\[
 \{\b K_{l}\},~~
 \{\b\tau^{s}_{r}=s\b t_{s}(h\b K_{s+r-1}+\b K_{s+r-2})+\b \sigma_{r}\}
\]
are symmetries of equation \eqref{eq-s},
and $\partial_{\b x}$ is a Noether operator of \eqref{eq-s} (i.e. $\partial_{\b x}^{-1}$
maps symmetries to conserved covariants for \eqref{eq-s}),
\be
 \{\b \gamma_{l}=\partial^{-1}_{\b x} \b K_l\},~~
 \{\b\vartheta^{s}_{r}=\partial^{-1}_{\b x} \b\tau^{s}_{r}=s\b t_{s}(h\b \gamma_{s+r-1}+\b \gamma_{s+r-2})+\b \omega_{r}\}
\ee
are conserved covariants of equation \eqref{eq-s}.
Since both $\{\b \gamma_{l}\}$ and $\{\b\vartheta^{s}_{r}\}$ are gradient functions
and their potentials are respectively $\{\b H_l\}$ and $\{\b I^s_r\}$
defined in \eqref{diKP:cq},
both $\{\b H_l\}$ and $\{\b I^s_r\}$ are conserved quantities of equation \eqref{eq-s}
thanks to Proposition \ref{P:2-2}.
In addition, obviously, $\b H_l$ is a conserved quantity of the whole isospectral D$\Delta$KP hierarchy \eqref{diKP:hie}
because $\b K_l$ is a symmetry of the whole isospectral  hierarchy.

Next, Let us prove the following relation:
\bse\label{diKP:alg-cc}
\begin{align}
&\{\b H_l,\b H_r\}=0,\label{diKP:cc-1}\\
&\{\b H_l,\b J_r\}=l\,(h\b H_{l+r-1}+\b H_{l+r-2}),\label{diKP:cc-2}\\
&\{\b J_l,\b J_r\}=(l-r)(h\b J_{l+r-1}+\b J_{l+r-2}).\label{diKP:cc-3}
\end{align}
\ese
In fact,
\[
\{\b H_l,\b H_r\} =\Bigl(\frac{\delta\b H_l}{\delta \b u},\,\partial_{\b x}\frac{\delta\b H_r}{\delta \b u}\Bigr)
=(\b \gamma_l,\,\partial_{\b x}\b \gamma_r)=(\b\gamma_l,\b K_r)
 =\b H_l'[\b K_r]=\b H_l'[\b u_{\b t_r}]=\frac{d\b H_l}{d{\b t_r}}.
\]
Since $\b H_l$ is a conserved quantity of the whole isospectral D$\Delta$KP hierarchy,
we know that $\frac{d \b H_l}{d {\b t_r}}=0$ and consequently $\{\b H_l,\b H_r\}=0$.
This also means the Hamiltonians $\{\b H_l\}$  of equation \eqref{eq-s} are involutive  w.r.t. the Poisson bracket $\{\cdot,\,\cdot\}$.
To derive \eqref{diKP:cc-2}, let us look at the relation
\[\llbracket \b K_l,\b\sigma_r \rrbracket=l\,(h\b K_{l+r-1}+\b K_{l+r-2}).\]
On one hand,
\[\partial_{\b x}^{-1}\llbracket \b K_l,\b\sigma_r \rrbracket
=\partial_{\b x}^{-1}(\partial_{\b x}\b\gamma_l'\b\sigma_r-\b\sigma_r'\partial_{\b x}\b\gamma_l)
=\b\gamma_l'^*\b\sigma_r+\b\sigma_r'^*\b\gamma_l=\mathrm{grad}(\b\gamma_l,\b\sigma_r),
\]
where we have made use of $\b\gamma_l'=\b\gamma_l'^*$, \eqref{eq-sg-m} and Lemma \ref{lem:grad}.
On the other hand,
\[l\,\partial_{\b x}^{-1}(h\b K_{l+r-1}+\b K_{l+r-2})=l\,(h\b \gamma_{l+r-1}+\b \gamma_{l+r-2}).\]
Thus we have
\[
(\b\gamma_l,\b\sigma_r)=l\,(h\b H_{l+r-1}+\b H_{l+r-2}).
\]
Meanwhile, noting that
\[
\{\b H_l,\b J_r\} =\Bigl(\frac{\delta\b H_l}{\delta \b u},\,\partial_{\b x}\frac{\delta\b J_r}{\delta \b u}\Bigr)
=(\b \gamma_l,\,\partial_{\b x}\b \omega_r)=(\b\gamma_l,\b \sigma_r),
\]
we immediately get \eqref{diKP:cc-2}.
\eqref{diKP:cc-3} can be proved similarly. From the relation
\[\llbracket \b \sigma_l,\b\sigma_r \rrbracket=(l-r)(h\b \sigma_{l+r-1}+\b \sigma_{l+r-2})\]
we have
\[
(\b\omega_l,\b\sigma_r)=(l-r)(h\b J_{l+r-1}+\b J_{l+r-2}).
\]
Besides,
\[
\{\b J_l,\b J_r\} =\Bigl(\frac{\delta\b J_l}{\delta \b u},\,\partial_{\b x}\frac{\delta\b J_r}{\delta \b u}\Bigr)
=(\b \omega_l,\,\partial_{\b x}\b \omega_r)=(\b\omega_l,\b \sigma_r).
\]
A combination of the above two formulae yields \eqref{diKP:cc-3}.

In the final step, the   relation \eqref{diKP:alg cq} can easily be verified
by using the algebra \eqref{diKP:alg-cc}.
Obviously,  \eqref{dKP:alg}, \eqref{diKP:alg sym}, \eqref{diKP:alg cq} and \eqref{diKP:alg-cc} are of same structures.
We complete the proof.
\end{proof}

\section{Continuum limits}\label{S:5}

\subsection{Backgrounds}

Let us write the KP equation and the D$\Delta$KP equation below,
\be
u_{t_3}=\frac{1}{4}u_{xxx}+3uu_x+\frac{3}{4}\partial_x^{-1}u_{yy},
\label{KP-s5}
\ee
\begin{equation}
\b u_{\b t_2}=(1+2\Delta^{-1})\b u_{\b x\b x}-2h^{-1}\b u_{\b x}+2\b u\b u_{\b x}.
\label{DDKP-s5}
\end{equation}
Following Miwa's transformation, or in practice, comparing exponential parts in the solution of these two equations,
one can introduce coordinates relation
\be
\label{Miwa trans}
x=\b x+\tau, ~~y=\b t_{2}-\frac{h}{2}\tau,~~t_3=\frac{h^{2}}{3}\tau, ~~~
\mathrm{with}~  \tau=nh.
\ee
The continuum limit is then conducted through replacing $\b u$ by $hu$ and
taking $n\to \infty$ and $h\to 0$ simultaneously.
The result is that the KP equation \eqref{KP-s5} appears as the leading term of the D$\Delta$KP equation \eqref{DDKP-s5}.
Similar relationship exists in non-commutative case\cite{LNT-ncKP}.

However, the continuum limit \eqref{Miwa trans} does not fit the whole D$\Delta$KP hierarchy.
It also breaks both basic algebraic structures and the Hamiltonian structure of the D$\Delta$KP equation.
In fact, to keep the Hamiltonian structure in a continuum limit, one at least needs $\b t_m\propto t_m$.
We need a new scheme for continuum limits.

\subsection{Plan for continuum limit}
\label{S:5.2}

Our plan for continuum limit is as following,
\begin{itemize}
\item{$n\rightarrow\infty$ and $h\rightarrow0$ simultaneously such that $nh$ is finite.}
\item{Introduce auxiliary continuous variable\footnote{In fact, we can take $\tau=\tau_0+nh$ with constant $\tau_0$.
Here we take $\tau_0=0$ for convenience and without loss of generality.}
  \begin{equation}
  \tau=nh,
  \label{tau}
  \end{equation}
  and thus, function $f(n+j)$ is mapped to $f(\tau+jh)$.}
\item{Define  coordinates relation
  \begin{equation}
  x=\b x+\tau,~~y=-\frac{1}{2}h\tau,~~t_m=\b t_m,
  \label{coord-rela}
  \end{equation}
  based on which one has
  \begin{equation}
  \partial_{\b x}=\partial_x, ~~\partial_{\tau} =\partial_x-\frac{1}{2}h\partial_y,~~ \partial_{\b t_m}=\partial_{t_m}.
  \end{equation}
  }
\item{ Define functions relation
  \begin{subequations}
  \begin{align}
  & \b u_0(n, \b x, \b t_m)= \b u(n, \b x, \b t_m)=h\,u(x, y, t_m),\\
  & \b u_j(n, \b x, \b t_m)= u_{j+1}(x, y, t_m),~~(j=1,2,\cdots).
  \end{align}
  \label{ub-u-rela}
  \end{subequations}
  }
\end{itemize}

\subsection{Pseudo-difference operator and D$\Delta$KP equation}
\label{S:5.3}

Under the continuum limit plan given in the above subsection, the pseudo-difference operator $\b L$
and pseudo-differential operator $L$ satisfy
\be
\b L=L+O(h).
\label{Lb-L}
\ee
In fact, acting $\Delta$ on a test function $f(n)$ and making use of Taylor expansion, one finds
\begin{align}
\Delta &=h \partial_{\tau}+\frac{1}{2!}h^2\partial_{\tau}^2+\frac{1}{3!}h^3\partial_{\tau}^3+O(h^4) \nonumber \\
&=h\partial_x+\frac{h^2 }{2}(\partial_x^2-\partial_y)+\frac{h^3 }{6}(\partial_x^3-3\partial_x\partial_y)+O(h^4),
\label{Delta-1}
\end{align}
and further,
\bse
\begin{align}
\Delta^{-1}&=h^{-1}\partial_x^{-1}+\left(\frac{1}{2}\partial_x^{-2}\partial_y
-\frac{1}{2}\right)+h\left(\frac{1}{2}\partial_x^{-3}\partial_y^2+\frac{1}{12}\partial_x\right)+O(h^2),\\
\Delta^{-2}&=h^{-2}\partial_x^{-2}+h^{-1}\left(\partial_x^{-3}\partial_y-\partial_x^{-1}\right)
+\left(\frac{3}{4}\partial_x^{-4}\partial_y^2-\frac{1}{2}\partial_x^{-2}\partial_y-\frac{5}{12}\right)+O(h),\\
&~~\cdots\cdots. \nonumber
\end{align}
\ese
Thus it is clear that
\be
h^{-j}\Delta^{j} =\partial_x^{j}+O(h),~~ j\in \mathbb{Z}.
\ee
Making use of this together with the relation \eqref{ub-u-rela} one immediately reaches to \eqref{Lb-L}.

Let us have a look at some lower order flows. In the continuum limit designed in Sec.\ref{S:5.2}, we find
\bse
\begin{align}
\b K_1=&h u_x=h\, K_1, \\
\b K_2=&h u_y+O(h^2)=h\, K_2+O(h^2), \\
\b K_3=&h\biggl(\frac{1}{4}u_{xxx}+3uu_x+\frac{3}{4}\partial_x^{-1}u_{yy}\biggr)+O(h^2)=h\, K_3+O(h^2).
\end{align}
\ese
It is not the so-called D$\Delta$KP equation $\b u_{\b t_2}=\b K_2$ but the next member in the D$\Delta$KP hierarchy,
i.e. $\b u_{\b t_3}=\b K_3$ that goes to the continuous KP equation $u_{t_3}=K_3$ in our continuum limit.

For the first three  non-isospectral flows, we find
\bse
\begin{align}
\b\sigma_1&=h(2yK_{1})+O(h)=h\, \sigma_1+O(h^2), \label{sigma1-b-nb}\\
\b\sigma_2&=h(2yK_{2}+xK_1+2u)+O(h^2)=h\, \sigma_2+O(h^2), \\
\b\sigma_3&=h(2yK_{3}+xK_2+2\partial_x^{-1}u_{y}-u_x)+O(h^2)=h\, \sigma_3+O(h^2).
\end{align}
\ese
Let us, taking \eqref{sigma1-b-nb} as an example, explain how the variable $y$ appears. In fact,
\begin{align*}
\b \sigma_1= h\b x \b u_{\b x} &=h^2(x-\tau)u_x=h^2 xu_x+2h y u_x\\
                               &=h\,\sigma_1+O(h^2).
\end{align*}

In brief, we have seen that, in our continuum limit, the first three D$\Delta$KP isospectral and non-isospectral flows
go to their continuous counterparts and the leading terms are of $O(h)$.

\subsection{Degrees}

In order to investigate the continuum limit of the whole  D$\Delta$KP hierarchies together with their
integrable properties, let us introduce \textit{degrees} for functions (cf.\cite{ZC-SAM-2010-II}).

\begin{defn}
Under the plan described in Sec.\ref{S:5.2}, a function $\b f(n,\b x, \b t_m)$ (or an operator $\b P(\b u, \Delta)$) can be expanded into
a series in terms of $h$,
where the order of the leading term is called the \textit{degree} of $\b f(n,\b x, \b t_m)$,
denoted by $\mathrm{deg}~\b f$.
\end{defn}

By this definition and previous discussion, we have
\begin{subequations}
\begin{align}
&\mathrm{deg}\, \b L=0,\\
&\mathrm{deg}\, \Delta^{j}=j,~~ j\in \mathbb{Z},\\
&\mathrm{deg}\, \b u=1,~~\mathrm{deg}\, \b u_j=0,~~(j=1,2,\cdots),
\end{align}
\end{subequations}
and
\[\mathrm{deg}\, \b K_j=1,~~ \mathrm{deg}\, \b \sigma_j=1,~~ (j=1,2,3).\]

Hereafter in this paper, by continuum limit we mean the one we designed in Sec.\ref{S:5.2},
without any confusion.
Let us first give some properties about degrees of functions and operations.

\begin{prop}
\label{P:5-1}
For the functions $\b f(\b u), ~\b g(\b u)$, it holds that
\begin{subequations}
\begin{align}
& \mathrm{deg}\,\b f\cdot \b g =\mathrm{deg}\, \b f+\mathrm{deg}\,\b g,\\
& \deg(\b f+\b g)\geq \mathrm{min}\{\deg \b f, \, \deg \b g\}. \label{deg-fg}
\end{align}
\end{subequations}
\end{prop}

\begin{prop}
\label{P:5-2}
For the functions $\b f(\b u)$ and $\b g(\b u)$ satisfying $\b f(\b u)|_{\b u=0}=0$ and $\b g(\b u)|_{\b u=0}=0$,
suppose that in continuum limit
\[\b f(\b u)=f(u)h^i+O(h^{i+1}),~~~
\b g(\b u)=g(u)h^j+O(h^{j+1}),
\]
i.e.
\[\deg \b f= i,~~\deg \b g=j.\]
It then holds that
\begin{align}
&\llbracket \b f(\b u),\b g(\b u)\rrbracket_{\b u}=\llbracket f(u),g(u)\rrbracket_{u}\, h^{i+j-1}+O(h^{i+j}),\label{f-g-deg1}\\
&\deg \llbracket \b f(\b u),\b g(\b u)\rrbracket_{\b u} \, \geq \deg\b f(\b u)+\deg\b g(\b u)-1.\label{f-g-deg2}
\end{align}
Here the subscripts $\b u$ and $u$ indicate the Lie brackets $\{ \llbracket \cdot,\,\cdot \rrbracket\}$
are defined based on the G\^ateaux derivatives w.r.t. $\b u$ and $u$, respectively.
\end{prop}
\begin{proof}
Noting that $\b u=h u$, we have
\begin{subequations}
\begin{align}
\b f'[\b g]&=\frac{d}{d\varepsilon}\b f(\b u+\varepsilon \b g(\b u))|_{\varepsilon=0}\nonumber\\
& = \frac{d}{d\varepsilon}\b f(h u+\varepsilon (g(u)h^j+O(h^{j+1}))) |_{\varepsilon=0} \nonumber \\
&= \frac{d}{d\varepsilon}\b f(h (u+\varepsilon (g(u)h^{j-1}+O(h^{j}))))|_{\varepsilon=0} \nonumber\\
&= \frac{d}{d\varepsilon}\Bigl(f(u+\varepsilon (g(u)h^{j-1}+O(h^{j})))h^i+\cdots\Bigr)\Bigr|_{\varepsilon=0}\nonumber \\
&=f'[g]h^{i+j-1}+O(h^{i+j}). \label{fg}
\end{align}
Similarly,
\be
 \b g'[\b f]=g'[f]h^{i+j-1}+O(h^{i+j}),
\ee
\end{subequations}
which, together with \eqref{fg}, yields \eqref{f-g-deg1}.
\eqref{f-g-deg2} is correct in light of \eqref{deg-fg}.
\end{proof}

\begin{prop}
\label{P:5-2+1}
If in continuum limit,
\[\b f(\b u)=f(u)h^i+O(h^{i+1}),~~~
\b g(\b u)=g(u)h^j+O(h^{j+1}),
\]
then
\begin{subequations}
\begin{align}
& (\b f(\b u),\b g(\b u))=(f(u),g(u))h^{i+j}+O(h^{i+j+1}),\label{eq-deg-inner}\\
& \deg (\b f(\b u),\b g(\b u))=\deg \b f(\b u)+\deg \b g(\b u).
\end{align}
\end{subequations}
Here on l.h.s. and r.h.s. of \eqref{eq-deg-inner} the inner products are defined
as \eqref{dKP:inn prod} for semi-discrete case and \eqref{def:inn prod} for continuous case, respectively.
This proposition also means that the degree of the semi-discrete inner product  \eqref{dKP:inn prod} is zero.
\end{prop}
\begin{proof}
First,
\begin{align*}
(\b f(\b u),\b g(\b u))&=\frac{h^2}{2}\sum_{n=-\infty}^{+\infty}\int_{-\infty}^{+\infty}\b f(\b u)\b g(\b u)\,\mathrm{d}\b x\\
&=\frac{h^2}{2}\sum_{n=-\infty}^{+\infty}\int_{-\infty}^{+\infty}( f(u) g(u)h^{i+j}+O(h^{i+j+1}))\,\mathrm{d}\b x\\
&=\frac{h}{2}\int_{-\infty}^{+\infty}\int_{-\infty}^{+\infty}( f(u) g(u)h^{i+j}+O(h^{i+j+1}))\,d\b x \mathrm{d}\tau.
\end{align*}
Next, from the coordinates transformation \eqref{coord-rela} we have the Jacobian
\[J=\frac{\partial(\b x,\tau)}{\partial(x,y)}=-\frac{2}{h}.\]
Then we have
\begin{align*}
(\b f(\b u),\b g(\b u))&=\frac{h}{2}\int_{-\infty}^{+\infty}\int_{-\infty}^{+\infty}( f(u) g(u)h^{i+j}+O(h^{i+j+1}))\,|J|\mathrm{d} x \mathrm{d} y\\
&= \int_{-\infty}^{+\infty}\int_{-\infty}^{+\infty}( f(u) g(u)h^{i+j}+O(h^{i+j+1}))\, \mathrm{d} x \mathrm{d} y\\
&=(f(u),g(u))h^{i+j}+O(h^{i+j+1}).
\end{align*}
This ends the proof.
\end{proof}

\begin{prop}
\label{P:5-2+2}
In  continuum limit if
\[\b \gamma (\b u)=\frac{\delta \b H(\b u)}{\delta \b u}=\gamma(u)h^i+O(h^{i+1}),\]
then we have
\be
\deg \b H(\b u)=\deg \b \gamma (\b u) +1.
\ee
In addition, if $\gamma(u)$ is also a gradient function, we can define
\be
H(u)=\int^1_0(  \gamma(\lambda  u),  u)\mathrm{d}\lambda,
\label{eq-H(u)}
\ee
and then we have
\be
\b H(\b u)=H(u)h^{i+1}+O(h^{i+2}), ~~ \gamma (u)=\frac{\delta  H(u)}{\delta  u}.
\label{eq-bH-H}
\ee
\end{prop}
\begin{proof}
Following Proposition \ref{P:5-2+1} and noting that
\be
\b H(\b u)=\int^1_0(\b \gamma(\lambda \b u), \b u)\mathrm{d}\lambda
=h^{i+1}\int^1_0(  \gamma(\lambda  u),  u)\mathrm{d}\lambda +O(h^{i+2}),
\label{(5.23)}
\ee
one has
\[\deg \b H(\b u)=\deg \gamma (\b u)+\deg \b u= \deg \b \gamma (\b u) +1.\]
If $\gamma(u)$ is a gradient function, after defining $H(u)$ in \eqref{eq-H(u)},
from \eqref{(5.23)} we reach to \eqref{eq-bH-H}.
\end{proof}

\begin{prop}
\label{P:5-2+3}
Suppose that in  continuum limit
\[\b \gamma (\b u)=\frac{\delta \b H(\b u)}{\delta \b u}
=\gamma(u)h^i+O(h^{i+1}),~~~ \b \vartheta (\b u)=\frac{\delta \b I(\b u)}{\delta \b u}=\vartheta(u)h^j+O(h^{j+1}),\]
and both $\gamma(u)$ and $\vartheta(u)$ are still gradient functions.
Then, according to Proposition \ref{P:5-2+2} we have
\be \b H(\b u)=H(u)h^{i+1}+O(h^{i+2}),~~~ \b I(\b u)=I(u)h^{j+1}+O(h^{j+2})
\ee
with
$\gamma (u)=\frac{\delta  H(u)}{\delta  u},~ \vartheta (u)=\frac{\delta  I(u)}{\delta  u}$,
and further
\begin{subequations}
\begin{align}
& \{\b H(\b u),\b I(\b u)\}=\{H(u),I(u)\}h^{i+j}+O(h^{i+j+1}),\\
& \deg \{\b H(\b u),\b I(\b u)\}=\deg \b H(\b u)+\deg \b I(\b u)-2.\label{eq-deg-poisson}
\end{align}
\end{subequations}
\end{prop}
\begin{proof}
Following \eqref{eq-deg-inner} and Proposition \ref{P:5-2+2}, one has
\begin{align*}
\{\b H(\b u),\b I(\b u)\}=(\b \gamma(\b u),\partial_{\b x}\b \vartheta(\b u))
=&(\gamma(u),\partial_{x}  \vartheta( u))h^{i+j}+O(h^{i+j+1})\\
=&\{H(u),I(u)\}h^{i+j}+O(h^{i+j+1}),
\end{align*}
which also indicates the degree relation \eqref{eq-deg-poisson}.
\end{proof}

Besides, the following lemmas will be helpful for investigating the degrees of $\b A_m,~\b B_m,~\b K_j$ and $\b\sigma_j$.

\begin{lem}\label{lem:5-1}
Suppose that $\b W_m$ is a difference operator
\[\b W_m=\sum^{m}_{j=0} \b w_j(\b{\mathbf{u}}) \Delta^{m-j},~~~ \mathrm{with}~\b W_m|_{\b{\mathbf{ u}}=0}=0.\]
If $\b W_m$ satisfies
\[ [\b W_m, \b L]=0,\]
then $\b W_m=0$.
\end{lem}
\begin{proof}
Arrange the terms of $ [\b W_m, \b L]$ in terms of $\Delta$.
The highest order term reads  $(\Delta\b w_0)\Delta^{m+1}$, which indicates
$\Delta\b w_0=0$.
This yields $\b w_0=0$ due to $\b W_m|_{\b{\mathbf{ u}}=0}=0$.
Thus, in the remains the highest order term is $(\Delta\b w_1)\Delta^{m}$ which should be zero,
and then we get $\b w_1=0$ by integration in the light of  $\b W_m|_{\b{\mathbf{ u}}=0}=0$.
Repeating the procedure we will finally reach to $\b W_m=0$ and finish the proof.
\end{proof}

Similarly we can have
\begin{lem}\label{lem:5-2}
For the differential operator
\[W_m=\sum^{m}_{j=0} w_j({\mathbf{u}}) \partial^{m-j},~~~ \mathrm{with}~W_m|_{{\mathbf{ u}}=0}=0.\]
if
\[ [W_m,  L]=0,\]
then $ W_m=0$.
\end{lem}

Now let us present more results on degrees.

\begin{prop}\label{P:5-3}
For the difference operator
\[\b W_m=\sum^{m}_{j=0} \b w_j(\b{\mathbf{u}}) \Delta^{m-j},\]
we have
\be
\deg [\b W_m, \b L]\geq \deg \b W_m, \label{deg-WL-1}
\ee
and if $\b W_m|_{\b{\mathbf{ u}}=0}=0$,
then
\be
\deg [\b W_m, \b L]= \deg \b W_m,
\label{deg-WL-2}
\ee
\end{prop}
\begin{proof}
\eqref{deg-WL-1} holds by virtue of Proposition \ref{P:5-1} and the fact $\deg \b L=0$.
Let us prove \eqref{deg-WL-2}.
Suppose that
\[\deg \b W_m=s,\]
i.e.
\[\b W_m=p(\partial_x) h^s+ O(h^{s+1}),\]
where $p(\partial_x)$ is some differential operator polynomial and $p(\partial_x)\neq 0$.
Then one has
\[[\b W_m, \b L]= [p(\partial_x), L] h^s+O(h^{s+1})\]
with leading term $[p(\partial_x), L]$.
If
\be
\deg [\b W_m, \b L]> \deg \b W_m,
\label{deg-WL-3}
\ee
which means the leading term of $[\b W_m, \b L]$ has to be zero, i.e.
\be
[p(\partial_x), L]=0.
\label{deg-pL}
\ee
Noting that $\b W_m|_{\b{\mathbf{u}}=0}=0$ yields $p(\partial_x)|_{\b{\mathbf{u}}=0}=0$,
from \eqref{deg-pL} and Lemma \ref{lem:5-2} one has $p(\partial_x)=0$.
This is contradictory to  $\deg \b W_m=s$, which means the assumption \eqref{deg-WL-3} is not correct,
and consequently \eqref{deg-WL-2} holds.
\end{proof}

\begin{prop}
\label{P:5-4}
In continuum limit,
\begin{subequations}\label{deg-A}
\begin{align}
& \deg \b A_m=0,\label{deg-A-1}\\
& \b A_m =A_m+O(h).
\label{deg-A-2}
\end{align}
\end{subequations}
\end{prop}
\begin{proof}
First, $\b A_m$ can be written in the following form
\be
\b A_m= \frac{\Delta^m}{h^m}+\sum^{m}_{j=1}\b a_j\frac{\Delta^{m-j}}{h^{m-j}}
=\partial_x^m + O(h) + \sum^{m}_{j=1}\b a_j (\partial_x^{m-j}+O(h)).
\label{A-1}
\ee
Note that $\b A_m=(\b L^m)_{+}$.
That is to say, $\b a_j$ only contains shifted $\b u_s$ without any integration terms like $\Delta^{-1} \b u_s$.
That means $\deg \b a_j \geq 0$ for all $j=1,2,\cdots m$,
and therefore \eqref{deg-A-1} holds.

Next, one can write $\b A_m$ as
\[\b A_m= A_m^{(0)}+O(h).\]
From \eqref{A-1} we know that $A^{(0)}_m$ is a differential operator and $A^{(0)}_m|_{u=0}=\partial^m_x$.
Now, from $\b L_{\b t_m}=[\b A_m,\b L]$ we have
\[L_{t_m}=[A_m^{(0)},L]+O(h)\]
and taking $h\to 0$ it goes to
\[L_{t_m}=[A_m^{(0)},L].\]
Finally, noting that $A^{(0)}_m|_{u=0}=A_m|_{u=0}=\partial^m_x$, and
making use of Lemma \ref{lem:5-2}, we have $A^{(0)}=A_m$, i.e. \eqref{deg-A-2} holds.
\end{proof}

\begin{prop}
\label{P:5-5}
In continuum limit,
\begin{subequations}\label{deg-B}
\begin{align}
& \deg \b B_m=0,\label{deg-B-1}\\
& \b B_m =B_m+O(h).
\label{deg-B-2}
\end{align}
\end{subequations}
\end{prop}

\begin{proof}
In the light of Lemma \ref{lem:5-1}, $\b B_m$ can be written in the following form
\be
\b B_m=\b D_m +\b C_{m-2},
\label{B-1}
\ee
where $\b D_m=\b x h \b A_m+ (\b x+nh)\b A_{m-1}$, $\b C_{m-2}$ is a pure difference operator and  $\b C_{m-2}|_{\b{\mathbf{u}}=0}=0$.
Suppose that $\deg \b C_{m-2}=s$, i.e.
\[\b C_{m-2}=C_{m-2}^{(0)} h^s+O(h^{s+1}).\]
From \eqref{dnKP:com-2} one has
\be
\b L_{\b t_m}=[\b D_m, \b L]+[\b C_{m-2}, \b L]+h \b L^m+\b L^{m-1},
\label{B-2}
\ee
where
$\deg [\b C_{m-2}, \b L]=\deg \b C_{m-2}=s$ due to $\b C_{m-2}|_{\b{\mathbf{u}}=0}=0$ together with Proposition \ref{P:5-3},
and the rest terms in \eqref{B-2} altogether have degree zero.
If $s<0$, there must have $[C_{m-2}^{(0)}, L]=0$ which yields $C_{m-2}^{(0)}=0$ in light of Lemma \ref{lem:5-2}.
This is in contradiction with  the assumption $\deg \b C_{m-2}=s<0$,
and consequently we must have $s\geq 0$. Thus,  noting that $\deg \b D_m=0$, \eqref{deg-B-1} holds.

With \eqref{deg-B-1} in hand, we can write
\be
\b B_m= B^{(0)}_m +O(h),
\label{B-3}
\ee
where
\[B^{(0)}_m=2y A_m+x A_{m-1} +C_{m-2},\]
and $C_{m-2}$ is a differential operator independent of $h$ and satisfies $C_{m-2}|_{\mathbf{u}=0}=0$.
Substituting \eqref{B-3} into \eqref{dnKP:com-2} the leading term is
\[L_{t_m}=[B^{(0)}_m, L]+L^{m-1}. \]
Obviously, $B_m$ and $B^{(0)}_m$ satisfy same equation and have asymptotic condition
$$B_m|_{\mathbf{u}=0}=B^{(0)}_m|_{\mathbf{u}=0}=2y \partial^m_x+x\partial^{m-1}_x,$$
which gives $B^{(0)}_m=B_m$ in the light of Lemma \ref{lem:5-2}.
Therefore the relation \eqref{deg-B-2} holds as well.
\end{proof}

\begin{prop}\label{P:5-6}
In continuum limit,
\begin{subequations}\label{deg-K}
\begin{align}
& \deg \b K_m=1,\label{deg-K-1}\\
& \b K_m =h\,K_m+O(h^2).
\label{deg-K-2}
\end{align}
\end{subequations}
\end{prop}

\begin{proof}
We would like to first specify the following relation,
\be
\b A_1= A_1+\frac{h}{2}(A_2-\partial_y)+ O(h^2).
\label{A-A1}
\ee
This can be derived by substituting \eqref{Delta-1} and $\b u_0=\b u=h u$ into $\b A_1$.
Actually, to derive \eqref{deg-K} we need higher order expansions.
Let us write
\be
\b A_m=A_m+A^{(1)}_{m}h+O(h^2).
\label{A-Am}
\ee
Inserting \eqref{A-A1} and \eqref{A-Am} into the zero curvature representation \eqref{diKP:zcr}
one has
\begin{align*}
\b K_m=& \b A_{m,\b x}-[\b A_1, \b A_m]\\
      =& \frac{h}{2}(A_{m,y}-[A_2,A_m])+O(h^2)\\
      =& h\,K_m+O(h^2).
\end{align*}
Besides,
\eqref{deg-K} can also be proved from \eqref{K-Res} through
\[\b K_m-\Delta \,\underset{\Delta}{\mathrm{Res\,}}\b L^m
=(K_m-\partial \,\underset{\partial}{\mathrm{Res\,}} L^m)h+O(h^2).
\]
Thus we compete the proof.
\end{proof}

In a quite similar way, using \eqref{A-A1}, \eqref{dnKP:zcr} and expression
\be
\b B_m=B_m+B^{(1)}_{m}h+O(h^2),
\label{B-Bm}
\ee
we have
\begin{prop}\label{P:5-7}
In continuum limit,
\begin{subequations}\label{deg-sigma}
\begin{align}
& \deg \b \sigma_m=1,\label{deg-sigma-1}\\
& \b \sigma_m =h\,\sigma_m+O(h^2).
\label{deg-sigma-2}
\end{align}
\end{subequations}
\end{prop}

\subsection{Lax triads}

From the previous discussion we have known that
\begin{subequations}
\begin{align}
& \b L=L+O(h),\\
& \b A_1= A_1+\frac{h}{2}(A_2-\partial_y)+ O(h^2),\\
& \b A_m=A_m+A^{(1)}_{m}h+O(h^2),\\
& \b B_m=B_m+B^{(1)}_{m}h+O(h^2).
\end{align}
\end{subequations}

Substituting them into the Lax triads and their compatibility equations in Sec.\ref{S:4}
we immediately reach to the following results.

\begin{prop}\label{P:5-8}
For the isospectral D$\Delta$KP hierarchy we have
\begin{subequations}
\begin{align}
& \b L \phi-\eta \phi= L\phi-\eta \phi + O(h),\\
& \phi_{\b x}-\b A_1 \phi=\frac{h}{2}(\phi_y-A_2\phi)+O(h^2),\\
& \phi_{\b t_m}-\b A_{m} \phi=\phi_{t_m}-A_m \phi+O(h),
\end{align}
\end{subequations}
and
\begin{subequations}
\begin{align}
& \b L_{\b x}-[\b A_{1},\b L]=\frac{h}{2}(L_y-[A_2,L])+O(h^2),\\
& \b L_{\b t_{m}}-[\b A_{m},\b L] =L_{t_m}-[A_m, L]+O(h),\\
& \b A_{1,\b t_m}-\b A_{m,\b x}+[\b A_1,\b A_m]=\frac{h}{2}(A_{2,t_m}- A_{m, y}+[ A_2,  A_m])+O(h^2).
\end{align}
\end{subequations}
\end{prop}

\begin{prop}\label{P:5-9}
For the non-isospectral D$\Delta$KP hierarchy we have
\begin{subequations}
\begin{align}
& \b L \phi-\eta \phi= L\phi-\eta \phi + O(h),\\
& \phi_{\b x}-\b A_1 \phi=\frac{h}{2}(\phi_y-A_2\phi)+O(h^2),\\
& \phi_{\b t_m}-\b B_{m} \phi=\phi_{t_m}-B_m \phi+O(h),
\end{align}
\end{subequations}
and
\begin{subequations}\label{deg-lax-i}
\begin{align}
& \b L_{\b x}-[\b A_{1},\b L]=\frac{h}{2}(L_y-[A_2,L])+O(h^2),\\
& \b L_{\b t_{m}}-[\b B_{m},\b L]-h\b L^{m}-\b L^{m-1}=L_{t_m}-[B_m, L]-L^{m-1}+O(h),\\
& \b A_{1,\b t_m}-\b B_{m,\b x}+[\b A_1,\b B_m]=\frac{h}{2}(A_{2,t_m}- B_{m, y}+[A_2,  B_m])+O(h^2).
\end{align}
\end{subequations}
\end{prop}

\subsection{Symmetries and algebra deformation}\label{S:5.6}

We have shown that both isospectral D$\Delta$KP flows $\{\b K_m\} $ and
non-isospectral D$\Delta$KP flows $\{\b \sigma_m\}$ go to their continuous counterparts
in continuum limit designed in Sec.\ref{S:5.2}.
However, comparing their basic algebra structures \eqref{KP:alg} and \eqref{dKP:alg},
one can see that their basic structures are different.
In fact, this deformation in basic structures can be well understood with the help of degrees of flows.
Let us take \eqref{dKP:alg-2} and \eqref{KP:alg-2} as an example.
\eqref{dKP:alg-2} reads
\be
{\llbracket}\b K_{l},\b \sigma_{r}{\rrbracket} =l\,(h\b K_{l+r-1}+\b K_{l+r-2}),
\label{dKP:alg-22}
\ee
among the three terms of which
\[\deg {\llbracket}\b K_{l},\b \sigma_{r}{\rrbracket}=1,
~~\deg (h\b K_{l+r-1})=2,~~
\deg \b K_{l+r-2}=1.
\]
Noting that in continuum limit only the terms with the lowest degrees (i.e. leading terms) are remained,
and comparing degrees of each term of \eqref{dKP:alg-22} we have
\be
{\llbracket}  K_{l},  \sigma_{r}{\rrbracket} =l\,K_{l+r-2} ,
\label{KP:alg-22}
\ee
i.e. \eqref{KP:alg-2}.
Such degree analysis works as well as in understanding  the relationship of symmetries together with their algebras
in semi-discrete and continuous cases.
Let us conclude these relations in the following.

\begin{thm}
\label{P:5-10}
In  continuum limit,
the basic algebra structure \eqref{dKP:alg} of flows goes to \eqref{KP:alg},
symmetries given in \eqref{diKP:sym}
\begin{align*}
\{\b K_l\} \rightarrow \{ K_l\},~~~~\{\b \tau_r^s\} \rightarrow \{\tau_r^s\},
\end{align*}
and their basic structure \eqref{diKP:alg sym} goes to \eqref{iKP:alg sym}.
\end{thm}

\subsection{Hamiltonian structures and conserved quantities}\label{S:5.7}

Now let us investigate continuum limits of Hamiltonian structures and conserved quantities.

Since
\[\b K_m=h\,K_m+O(h^2),~~~ \partial_{\b x}= \partial_{x},\]
it is easy to have
\[\b K_m =\partial_{\b x} \b \gamma_m= h\,K_m+O(h^2)=h\,\partial_x \gamma_m+O(h^2),\]
i.e.
\[\b \gamma_m= h\,\gamma_m + O(h^2),\]
and $\gamma_m$ is still a gradient function.
Then, following Proposition \ref{P:5-2+3} we have
\[\b H_m= h^2 H_m + O(h^3).\]

We can conduct similar discussion for the non-isospectral case and get similar results.
To sum up, for Hamiltonian structures we have
\begin{prop}\label{P:5-11}
The continuum limit designed in Sec.\ref{S:5.2}
keeps the Hamiltonian structures of the equation \eqref{diKP:Ham} and \eqref{dnKP:Ham}, in which
\bse
\begin{align}
& \b \gamma_m= h\, \gamma_m + O(h^2),~~  \b H_m=   h^2 H_m + O(h^3),\\
& \b \omega_m= h\, \omega_m + O(h^2),~~  \b J_m=   h^2 J_m + O(h^3).
\end{align}
\ese
\end{prop}

Next we look at the basic algebraic structure \eqref{diKP:alg cq} composed by the conserved quantities $\{\b H_l\}$ and $\{\b I_r^s\}$.
We have seen that in continuum limit $\b \gamma_m(\b u)$ and $\b \omega_m(\b u)$
go to $\gamma_m(u)$ and $\omega_m(u)$ that are still gradient functions.
Noting that $\gamma_m (u)=\frac{\delta  H_m(u)}{\delta  u},~ \omega_m (u)=\frac{\delta  J_m(u)}{\delta  u}$,
it then follows from Proposition \ref{P:5-2+3} that
in continuum limit
\begin{align*}
&\{\b H_l,\b H_r\}=\{ H_l, H_r\}h^{2}+O(h^3),\\
&\{\b H_l,\b J_r\}=\{ H_l, J_r\}h^{2}+O(h^3),\\
&\{\b J_l,\b J_r\}=\{ J_l, J_r\}h^{2}+O(h^3).
\end{align*}
We use the same trick as in the previous subsection for symmetries.
By comparing degrees of both sides of the basic algebraic relation \eqref{diKP:alg-cc},
the leading terms give
\bse
\label{iKP:alg cc-11}
\begin{align}
&\{ H_l, H_r\}=0,\label{iKP:cc-11}\\
&\{ H_l, J_r\}=l\,H_{l+r-2},\label{iKP:cc-21}\\
&\{ J_l, J_r\}=(l-r) J_{l+r-2},\label{iKP:cc-31}
\end{align}
\ese
i.e. \eqref{iKP:alg cc}.
This also leads to the basic algebraic relation \eqref{iKP:alg cq}.
Let us conclude it in the following.
\begin{thm}
\label{P:5-12}
In the continuum limit in Sec.\ref{S:5.2}, we have
\[
\{\b H_l\} \rightarrow \{ H_l\},~~\{\b J_r\} \rightarrow \{J_r\},~~\{\b I_r^s\} \rightarrow \{I_r^s\},
\]
the  basic algebra structure \eqref{diKP:alg-cc} goes to \eqref{iKP:alg cc}
and the basic structure \eqref{diKP:alg cq} goes to \eqref{iKP:alg cq}.
\end{thm}

\subsection{Deformation of Lie algebras}
\label{S:5.8}

Now let us see something special of the obtained algebras.
The Lie algebra $\b{\mathbf{X}}$ spanned by the  D$\Delta$KP flows $\{ \b K_m\}$ and $\{\b \sigma_m\}$ with the basic structures \eqref{dKP:alg}
has generators $\{\b K_1, \b \sigma_1, \b \sigma_3\}$ w.r.t. the product $\j\cdot,\cdot \k$;
while the Lie algebra ${\mathbf{X}}$ spanned by the KP flows $\{ K_m\}$ and $\{ \sigma_m\}$  with the basic structures \eqref{KP:alg}
has generators $\{ K_1,  \sigma_1, \sigma_4\}$.
Obviously, the two algebras have different basic structures:
\eqref{KP:alg} is a neat centerless Kac-Moody-Virasoro structure but \eqref{dKP:alg} is not.
Now let us look at subalgebras.
${\mathbf{X}}$ has infinitely many subalgebras spanned by $\{ K_1, K_2,\cdots,K_j,  \sigma_1, \sigma_2\}$ for any $j\in \mathbb{Z}^+$;
for $\b{\mathbf{X}}$ it also has infinitely many subalgebras
spanned by  $\{\b K_1, \b K_2, \cdots, \b K_j, \b \sigma_1\}$ for any $j\in \mathbb{Z}^+$.
Moreover, by means of calculating degrees of flows
the deformation in the basic algebraic structures can be understood in continuum limit.
However, the continuum limit does not keep generators and subalgebras.
In fact, such discontinuity of Lie algebras of flows (or symmetries), also known as the contraction of algebras,
is not rare to see in some semi-discrete cases when they go to their continuous correspondences in continuous
limit\cite{HHLRW-JPA-2000,HHLW-TMP-2001,ZC-SAM-2010-II}.
Here, the spacing parameter $h$  acts as a contraction parameter that bring changes of basic algebraic structures.

Since the basic algebraic structures \eqref{diKP:alg-cc} for Hamiltonians, \eqref{diKP:alg cq} for conserved quantities and \eqref{dKP:alg} for flows
are the same, and
the basic structures \eqref{iKP:alg cc} for Hamiltonians, \eqref{iKP:alg cq} for conserved quantities and \eqref{KP:alg} for flows
are also same,
they have same deformations.

\section{Conclusions}

We have discussed integrable properties of the D$\Delta$KP hierarchy,
including symmetries, Hamiltonian structures and conserved quantities.
The obtained results are
isospectral and non-isospectral  D$\Delta$KP flows and their Lie algebra,
two sets of symmetries of the isospectral hierarchy and their Lie algebra,
Hamiltonian structures of isospectral and non-isospectral hierarchies,
Lie algebra of the Hamiltonians,
two sets of conserved quantities of the isospectral hierarchy and their Lie algebra,
and all these Lie algebras have same basic algebraic structures.
To achieve these, we introduced Lax triads as our starting point.
In this approach we consider the spatial variable $\bar{x}$ ($y$ for the KP system) as a new  independent variable
that is completely independent of the temporal variable $\bar{t}_1$ ($t_2$ for the KP system).
Such a separation of spatial and temporal variables
not only enables us to derive master symmetries as non-isospectral flows
but also provides simple zero curvature representations for both
isospectral and non-isospectral flows,
which leads to a Lie algebra with recursive structures of these flows.
Compared with the traditional treatments,
we believe that the Lax triad approach  would be more reasonable
in the study of (2+1)-dimensional systems related to pseudo-difference operators and pseudo-differential operators.
Besides, explicit recursion operators might exist and be used to investigate integrable (2+1)-dimensional systems\cite{SF-CMP-1988-I,FS-CMP-1988-II},
which is absent in discrete case.

Continuum limit acts as a bridge to connect discrete and continuous integrable systems.
However, such connections usually are hidden behind integrable discretization\cite{Hirota-JPSJ-1977-I,Hirota-JPSJ-1977-II,Hirota-JPSJ-1977-III,
MP-CMP-1996,MP-RMP-1998,MP-JPA-1998,ZC-SAM-2010-II}.
It is not easy to find out a uniform continuum limit to connect both equations and their integrable properties,
and sometimes combinatorics are used.
In the paper we designed a continuum limit that connects the D$\Delta$KP and KP hierarchies.
The continuum limit has  been shown  to keep their Lax triads, zero curvature representations, Hamiltonian structures
(for both isospectral and non-isospectral cases), symmetries and conserved quantities.
We defined and made use of \textit{degrees} of some elements
to analyze continuum limits.
By calculating and comparing degrees
the deformation in the basic algebraic structures can be understood in continuum limits.
We also want to emphasize that in our continuum limit
the traditional D$\Delta$KP equation $\b u_{\b t_2}=\b K_2$ goes to the linear equation $u_{t_2}=u_y$
rather than the KP equation.
It turns out that the next member $\b u_{\b t_3}=\b K_3$  corresponds to the continuous KP equation.

The pseudo-difference operator $\b L$ is not a unique means for investigating the D$\Delta$KP hierarchy.
In a series of papers \cite{WCN-PA-1986,WC-PLA-1987,WC-PA-1988-II,WC-PA-1988-III} the discrete KP equation
together with related continuum limits, conserved quantities,
Hamiltonian structures and semi-discrete KP hierarchies were investigated starting from the
so-called direct linearization approach.
In their approach fully discrete KP is a starting point,
infinitely many conserved quantities were derived from a time-independent scattering data,
and semi-discrete hierarchy were generated in continuum limit by defining an infinite number of
continuous temporal variables.
Here we have given more conserved quantities and more algebraic structures for the D$\Delta$KP hierarchy.
The integrable master symmetries played important roles in our paper,
and in the continuum limit we have fixed time variables so that the continuum limit keeps Hamiltonian structures for the whole hierarchies.

\vskip 10pt
\subsection*{Acknowledgments}
The authors are very grateful to Prof. Deng-yuan Chen for discussion.
We also sincerely thank Prof. Frank W. Nijhoff for kindly pointing out some
important references. This project is  supported by the NSF of China
(No. 11071157), SRF of the DPHE of China (No. 20113108110002) and
Shanghai Leading Academic Discipline Project (No. J50101).

\vskip 20pt

\begin{appendix}
\section{Formulae on G\^{a}teaux derivatives}\label{A:A}

We collect some formulae of G\^{a}teaux derivatives that are often used.
For convenience we take $F=F(x,\{v^{(j)}\})$ as an example where $v=v(t,x)$ and $v^{(j)}=\partial^{j}_x v$.
These formulae are
\begin{align*}
& F'[g]=\sum_j\frac{\partial F}{\partial v^{(j)}} \partial^{j}_x g,\\
& \partial_x F=\t \partial_x F +F'[v_x],\\
& \t \partial_x (F'[g])=(\t \partial_x F)'[g]+F'[\t \partial_x g],\\
& (F'[a])'[b]=(F')'[b]\circ[a]+F'\circ a'[b],\\
& (F')'[a]\circ [b]=(F')'[b]\circ[a],\\
& F'\j a,b\k=(F'[a])'[b]-(F'[b])'[a],\\
& \partial_x (F'[g])=(\partial_x F)'[g].
\end{align*}
The first formula can be used to prove others.

\section{Discussion of conserved quantities}\label{A:B}

Based on Proposition \ref{P:2-2-0}, using symmetries and conserved covariants one can construct conserved quantities via inner product
$(\cdot, \cdot)$.
Conserved quantities can also be constructed through gradient functions and Proposition \ref{P:2-1}.
It is necessary to investigate the relationship of these conserved quantities derived from different ways.
For the isospectral KP hierarchy, we have

\begin{thm}\label{T:B-1}
For each equation
\be
u_{t_{s}}=K_{s}
\label{uts}
\ee
in the isospectral KP hierarchy \eqref{iKP:hie}, we have
\bse\label{iKP:scc}
\begin{align}
& (K_{l},\gamma_{r}) = 0,\label{KP:scc-1}\\
& (K_{l},\vartheta^{s}_{r})= -(\tau^{s}_{l},\gamma_r)=-l\,H_{l+r-2},\label{KP:scc-2}\\
& (\tau^{s}_{l},\vartheta^{s}_{r}) = -(l-r)I^{s}_{l+r-2},\label{KP:scc-3}
\end{align}
\ese
where $l,r,s\geq 1$, $H_m$ and $ I^{s}_{m}$ are conserved quantities defined in Theorem \ref{thm:iKP Ham} and Theorem \ref{T:3-3},
\be
\{K_{l}\},~~
\{\tau^{s}_{r}=st_sK_{s+r-2}+\sigma_r\}
\ee
are symmetries of equation \eqref{uts},
\be
\{\gamma_{l}=\partial^{-1}K_l\},~~
\{\vartheta^{s}_{r}=\partial^{-1}\tau^{s}_{r}= st_s \gamma_{s+r-2}+\omega_r\}
\ee
are conserved covariants  of   \eqref{uts}.
Note that we set $H_0=I^{s}_{0}=0$ and $K_0=\tau_0^s=0$.
\end{thm}

We skip the proof and a similar proof will be given in the next theorem.

\begin{thm}\label{T:B-2}
For each equation
\be
\b u_{\b t_{s}}=\b K_{s}
\label{b-uts}
\ee
in the isospectral D$\Delta$KP hierarchy \eqref{diKP:hie}, we have
\bse\label{idKP:scc}
\begin{align}
& (\b K_{l},\b\gamma_{r})=0,\label{idKP:scc-1}\\
& (\b K_{l},\b \vartheta^{s}_{r})=-(\b \tau^{s}_{l},\b \gamma_r)=-l\,(h\b H_{l+r-1}+\b H_{l+r-2}),\label{idKP:scc-2}\\
& (\b \tau^{s}_{l},\b \vartheta^{s}_{r})=-(l-r)(h \b I^s_{l+r-1}+ \b I^{s}_{l+r-2}),\label{idKP:scc-3}
\end{align}
\ese
where $l,r,s\geq 1$, $\b H_m$ and $\b I^{s}_{m}$ are conserved quantities defined in Theorem \ref{thm:diKP Ham} and Theorem \ref{T:4-4},
\be
 \{\b K_{l}\},~~
 \{\b\tau^{s}_{r}=s\b t_{s}(h\b K_{s+r-1}+\b K_{s+r-2})+\b \sigma_{r}\}
\ee
are symmetries of equation \eqref{uts},
\be
 \{\b \gamma_{l}=\partial^{-1}_{\b x} \b K_l\},~~
 \{\b\vartheta^{s}_{r}=\partial^{-1}_{\b x} \b\tau^{s}_{r}=s\b t_{s}(h\b \gamma_{s+r-1}+\b \gamma_{s+r-2})+\b \omega_{r}\}
\ee
are conserved covariants  of   \eqref{b-uts}.
Here we set $\b H_0=\b I^{s}_{0}=0$ and $\b K_0=\b \tau_0^s=0$.

\end{thm}

\begin{proof}
We only prove \eqref{idKP:scc-3}, and others are similar.
Noting that the relation \eqref{diKP:alg sym-c}, i.e.
\be
 {\llbracket}\b \tau^{s}_{l},\b \tau^{s}_{r}{\rrbracket}=(l-r)(h\b\tau^{s}_{l+r-1}+\b\tau^{s}_{l+r-2})
 \label{b-9}
\ee
from l.h.s. we have
\begin{align*}
 \partial^{-1}_{\b x} {\llbracket}\b \tau^{s}_{l},\b \tau^{s}_{r}{\rrbracket}
  & =  \partial^{-1}_{\b x}( {\b \tau^{s\prime}_{l}}[\b\tau^{s}_{r}]-{\b \tau^{s\prime}_{r}}[\b\tau^{s}_{l}])\\
  & =  \partial^{-1}_{\b x}( {\b\tau^{s \prime}_{l}}\partial_{\b x}\b\vartheta^{s}_{r}-\partial_{\b x}\b\vartheta^{s\prime}_{r}\b \tau^{s}_{l}).
\end{align*}
Then, in the light of Corollary \ref{C:4-3} we get ${\b\tau^{s\prime}_{l}}\partial_{\b x} =-\partial_{\b x} {\b\tau^{s\prime \, *}_{l}}$ and then
\begin{align*}
 \partial^{-1}_{\b x} {\llbracket}\b \tau^{s}_{l},\b \tau^{s}_{r}{\rrbracket}
  & = - \partial^{-1}_{\b x}(  \partial_{\b x} {\b\tau^{s \prime *}_{l}}\b\vartheta^{s}_{r}+\partial_{\b x} {\b\vartheta^{s\prime \, *}_{r}} \b\tau^{s}_{l})\\
  & = - ({\b\tau^{s\prime \, *}_{l}} \b\vartheta^{s}_{r}+{\b\vartheta^{s \prime \, *}_{r}}\b\tau^{s}_{l})\\
  & = - \mathrm{grad} (\b\tau^{s}_{l},\b\vartheta^{s}_{r}),
\end{align*}
where we have made use of $\b\vartheta^{s \prime}_{r}=\b\vartheta^{s \prime *}_{r}$ and Lemma \ref{lem:grad}.
Meanwhile, from the r.h.s. of \eqref{b-9} we have
\[\partial^{-1}_{\b x}(l-r)(h\b\tau^{s}_{l+r-1}+\b\tau^{s}_{l+r-2})
=(l-r)(h\b\vartheta^{s}_{l+r-1}+\b\vartheta^{s}_{l+r-2}).\]
Thus, recovering potentials from the above two formulae we have the relation
\[(\b\tau^{s}_{l},\b\vartheta^{s}_{r})=-(l-r)(h\b I^{s}_{l+r-1}+\b I^{s}_{l+r-2})+c,\]
where $c$ is at most related to $\b t_s$ because $\b I^{s}_{m}$ is also defined through inner product.
Noting that both $(\b\tau^{s}_{l},\b\vartheta^{s}_{r})$ and $\b I^{s}_{m}$ are conserved quantities of equation \eqref{b-uts},
$c$ must be independent of $\b t_s$ and therefore it becomes trivial and  we can take $c=0$ without loss of generality.
Thus we reach to \eqref{idKP:scc-3}.
\end{proof}

With regard to the relationship of \eqref{idKP:scc} and \eqref{iKP:scc},
thanks to the results obtained in Sec.\ref{S:5}, we can conclude that
\begin{thm}
In the continuum limit designed in Sec.\ref{S:5.2}, the relation \eqref{idKP:scc} goes to \eqref{iKP:scc}.
\end{thm}

\end{appendix}


\vskip 20pt



{\small

\begin{thebibliography}{99}
\bibitem{OSTT-PTPS-1988}
    Y. Ohta, J. Satsuma, D. Takahashi, T. Tokihiro,
    An elementary introduction to Sato theory,
    Prog. Theor. Phys. Suppl., 94 (1988) 210-41.
\bibitem{DJM-book-2000}
    T. Miwa, M. Jimbo, E. Date,
    \textit{Solitons: Differential Equations, Symmetries and Infinite Dimensional Algebras},
    Cambridge University Press, Cambridge, 2000.
\bibitem{Dickey-book-2003}
    L.A. Dickey,
    \textit{Soliton Equations and Hamiltonian Systems}, 2nd ed.,
    World Scientific, Singapore, 2003.
\bibitem{MSS-JMP-1990}
    J. Matsukidaira, J. Satsuma, W. Strampp,
    Conserved quantities and symmetries of KP hierarchy,
    J. Math. Phys., 31 (1990) 1426-34.
\bibitem{CXZ-CSF-2003}
    D.Y. Chen, H.W. Xin, D.J. Zhang,
    Lie algebraic structures of (1+2)-dimensional Lax integrable systems,
    Chaos, Solitons \& Fractals, 15 (2003) 761-70.
\bibitem{SF-CMP-1988-I}
    P.M. Santini, A.S. Fokas,
    Recursion operators and bi-Hamiltonian structures in multidimensions (I),
    Commun. Math. Phys., 115 (1988) 375-419.
\bibitem{OF-PLA-1982}
    W. Oevel, B. Fuchssteiner,
    Explicit formulas for symmetries and conservation laws of the Kadomtsev-Petviashvili equation,
    Phys. Lett. A, 88 (1982) 323-7.
\bibitem{CLL-PD-1983}
    H.H. Chen, Y.C. Lee, J.E. Lin,
    On a new hierarchy of symmetries for the Kadomtsev-Petviashvili equation,
    Physica D, 9 (1983) 439-45.
\bibitem{Case-PNAS-1984}
    K.M. Case,
    A theorem about Hamiltonian systems,
    Proc. Natl. Acad. Sci. USA, 81 (1984) 5893-95.
\bibitem{Case-JMP-1985}
    K.M. Case,
    Symmetries of the higher-order KP equations,
    J. Math. Phys., 26 (1985) 1158-59.
\bibitem{CLB-JPA-1988}
    Y. Cheng, Y.S. Li, R.K. Bulough,
    Integrable nonisospectral flows associated with the Kadomtsev-Petviashvili equations in 2+1 dimensions,
    J. Phys. A: Math. Gen., 21 (1988) L443-9.
\bibitem{DJM-JPSJ-1982-II}
    E. Date, M. Jimbo, T. Miwa,
    Method for generating discrete soliton equations: II,
    J. Phys. Soc. Jpn., 51 (1982) 4125-31.
\bibitem{Miwa-PJA-1982}
    T. Miwa,
    On Hirota's difference equations,
    Proc. Jpn. Acad., 58A (1982) 9-12.
\bibitem{DJM-JPSJ-1982-I}
    E. Date, M. Jimbo, T. Miwa,
    Method for generating discrete soliton equations: I,
    J. Phys. Soc. Jpn., 51 (1982) 4116-24.
\bibitem{DJM-JPSJ-1983-III}
    E. Date, M. Jimbo, T. Miwa,
    Method for generating discrete soliton equations: III,
    J. Phys. Soc. Jpn.,  52 (1983) 388-93.
\bibitem{DJM-JPSJ-1983-IV}
    E. Date, M. Jimbo, T. Miwa,
    Method for generating discrete soliton equations: IV,
    J. Phys. Soc. Jpn., 52 (1983) 761-65.
\bibitem{DJM-JPSJ-1983-V}
    E. Date, M. Jimbo, T. Miwa,
    Method for generating discrete soliton equations: V,
    J. Phys. Soc. Jpn., 52 (1983) 766-71.
\bibitem{KVT-CSF-1997}
    S. Kanaga Vel, K.M. Tamizhmani,
    Lax pairs, symmetries and conservation laws of a differential-difference equation-Sato's approach,
    Chaos, Solitons \& Fractals, 8 (1997) 917-31.
\bibitem{K-DOC-1998}S. Kanaga Vel,
    On certain integrability aspects of differential-difference Kadomtsev-Petviashvili equation,
    PhD Thesis, Pondicherry University, India, 1998.
\bibitem{SZZC-MPLB-2010}
    X.L. Sun, D.J. Zhang, X.Y. Zhu, D.Y. Chen,
    Symmetries and Lie algebra of the differential-difference Kadomstev-Petviashvili hierarchy,
    Mod. Phys. Lett. B, 24 (2010) 1033-42.
\bibitem{Zhang-JSU-2005}
    D.J. Zhang,
    Conservation laws of the differential-difference KP equation,
    J. Shanghai Univ., 9 (2005) 206-9.
\bibitem{Fuch-PTP-1983}
    B. Fuchssteiner,
    Master symmetries, higher order time-dependent symmetries and conserved densities of nonlinear evolution equations,
    Prog. Theor. Phys., 70 (1983) 1508-22.
\bibitem{FF-PD-1981}
    B. Fuchssteiner, A.S. Fokas,
    Symplectic structures, their B\"acklund transformations and hereditary symmetries,
    Physica D, 4 (1981) 47-66.
\bibitem{Zhang-JPSJ-2003}
    D.J. Zhang, D.Y. Chen,
    Some general formulas in the Sato theory,
    J. Phys. Soc. Jpn., 72 (2003) 448-9.
\bibitem{CZ-JPA-1991}
    D.Y. Chen, H.W. Zhang,
    Lie algebraic structure for the AKNS system,
    J. Phys. A: Math. Gen., 24 (1991) 377-83.
\bibitem{Ma-JMP-1992}
    W.X. Ma,
    Lax representations and Lax operator algebras of isospectral and nonisospectral hierarchies of evolution equations,
    J. Math. Phys., 33 (1992) 2464-76.
\bibitem{MF-JMP-1999}
    W.X. Ma, B. Fuchssteiner,
    Algebraic structure of discrete zero curvature equations and master symmetries of discrete evolution equations,
    J. Math. Phys., 40 (1999) 2400-18.
\bibitem{TM-JPSJ-1999}
    K.M. Tamizhmani, W.X. Ma,
    Master symmetries from Lax operators for certain lattice soliton hierarchies,
    J. Phys. Soc. Jpn., 69 (2000) 351-61.
\bibitem{CZ-JMP-1996}
    D.Y. Chen, D.J. Zhang,
    Lie algebraic structures of (1+1)-dimensional Lax integrable systems,
    J. Math. Phys., 37 (1996) 5524-38.
\bibitem{Zhang-PLA-2006} D.J. Zhang, T.K. Ning, J.B. Bi, D.Y. Chen,
    New symmetries for the Ablowitz-Ladik hierarchies,
    Phys. Lett. A, 359(5), (2006), 458-66.
\bibitem{ZC-SAM-2010-I}
    D.J. Zhang, S.T. Chen,
    Symmetries for the Ablowitz-Ladik hierarchy: Part I. Four-potential case,
    Stud. Appl. Math., 125 (2010) 393-418.
\bibitem{ZC-JPA-2002}
    D.J. Zhang, D.Y. Chen,
    Hamiltonian structure of discrete soliton systems,
    J. Phys. A: Math. Gen., 35 (2002) 7225-41.
\bibitem{ZC-SAM-2010-II}
    D.J. Zhang, S.T. Chen,
    Symmetries for the Ablowitz-Ladik hierarchy: Part II. Integrable discrete nonlinear Schr\"odinger equations and discrete AKNS hierarchy,
    Stud. Appl. Math., 125 (2010) 419-43.
\bibitem{LNT-ncKP}C.X. Li, J.J.C. Nimmo, K.M. Tamizhmani,
    On solutions to the non-Abelian Hirota-Miwa equation and its continuum limits,
    Proc. R. Soc. A, 465, (2009) 1441-51.
\bibitem{HHLRW-JPA-2000}
    R. Hern\'andez Heredero, D. Levi, M.A. Rodr\'iguez, P. Winternitz,
    Lie algebra contractions and symmetries of the Toda hierarchy,
    J. Phys. A: Math. Gen., 33 (2000) 5025-40.
\bibitem{HHLW-TMP-2001}
    R. Hern\'andez Heredero, D. Levi, P. Winternitz,
    Symmetries of the discrete nonlinear Schr\"odinger equation,
    Theor. Math. Phys., 127 (2001) 729-37.
\bibitem{FS-CMP-1988-II}
    A.S. Fokas, P.M. Santini,
    Recursion operators and bi-Hamiltonian structures in multidimensions (II),
    Commun. Math. Phys., 116 (1988) 449-74.
\bibitem{Hirota-JPSJ-1977-I} R. Hirota,
    Nonlinear partial difference equations. I. A difference analogue of the Korteweg-de Vries equation,
    J. Phys. Soc. Jpn., 43 (1977) 1424-33.
\bibitem{Hirota-JPSJ-1977-II} R. Hirota,
    Nonlinear partial difference equations. II. Discrete-time Toda equation,
    J. Phys. Soc. Jpn., 43 (1977) 2074-8.
\bibitem{Hirota-JPSJ-1977-III} R. Hirota,
    Nonlinear partial difference equations. III. Discrete sine-Gordon equation,
    J. Phys. Soc. Jpn., 43 (1977) 2079-86.
\bibitem{MP-CMP-1996}
    C. Morosi, L. Pizzocchero,
    On the continuous limit of integrable lattices I. The Kac-Moerbeke system and KdV theory,
    Commun. Math. Phys., 180 (1996) 505-28.
\bibitem{MP-RMP-1998}
    C. Morosi, L. Pizzocchero,
    On the continuous limit of integrable lattices II. Volterra systems and $sp(N)$ theories,
    Rev. Math. Phys., 10 (1998) 235-70.
\bibitem{MP-JPA-1998}
    C. Morosi, L. Pizzocchero,
    On the continuous limit of integrable lattices III. Kupershmidt systems and $sl(N + 1)$ KdV theories,
    J. Phys. A: Math. Gen., 31 (1998) 2727-46.
\bibitem{WCN-PA-1986}
    G.L. Wiersma, H.W. Capel, F.W. Nijhoff,
    Linearizing integral transformation for the multicomponent lattice KP,
    Physica A, 138 (1986) 76-99.
\bibitem{WC-PLA-1987}
    G.L. Wiersma, H.W. Capel,
    Lattice equations, hierarchies and Hamiltonian structures: The Kadomtsev-Petviashvili equation,
    Phys. Lett. A, 124 (1987) 124-30.
\bibitem{WC-PA-1988-II}
    G.L. Wiersma, H.W. Capel,
    Lattice equations, hierarchies and Hamiltonian structures: II. KP-type of hierarchies on 2D lattices,
    Physica A, 149 (1988) 49-74.
\bibitem{WC-PA-1988-III}
    G.L. Wiersma, H.W. Capel,
    Lattice equations, hierarchies and Hamiltonian structures: III. The 2D Toda and KP hierarchies,
    Physica A, 149 (1988) 75-106.
\end{thebibliography}
}
\end{document}